\documentclass[a4paper,UKenglish,cleveref, autoref, thm-restate]{lipics-v2021}

\usepackage{tikz}
\usetikzlibrary{positioning, shapes.geometric, decorations.pathreplacing, calligraphy, calc, backgrounds}
\usepackage{comment,todonotes}
\usepackage{xspace}

\usepackage{bm} 

\newcommand{\vcew}{{\sc Vertex-coloring Edge-weighting}\xspace}

\newcommand{\prvcew}{{\sc Vertex-Coloring Pre-edge-Weighting}\xspace}
\newcommand{\lc}{{\sc List Coloring}\xspace}

\newcommand{\onetwothree}{1--2--3~Conjecture\xspace} 
\newcommand{\yes}{\textsf{Yes}\xspace}
\newcommand{\no}{\textsf{No}\xspace}
\newcommand{\XP}{\textsf{XP}\xspace}
\newcommand{\FPT}{\textsf{FPT}\xspace}
\newcommand{\WoneHard}{\textsf{W[1]-hard}\xspace}

\nolinenumbers

 \newcommand{\defproblem}[3]{
  \vspace{1mm}
\noindent\fbox{
  \begin{minipage}{0.96\textwidth}
   \begin{tabular*}{\textwidth}{@{\extracolsep{\fill}}lr} #1  
 \\ \end{tabular*}
  {\bf{Input:}} #2  \\
  {\bf{Question:}} #3
  \end{minipage}
  }
 \vspace{1mm}
}

\title{The Parameterized Complexity of Vertex-Coloring Edge-Weighting} 

\author{Shubhada Aute 
}{Department of Computer Science and Engineering, IIT Hyderabad, India}{cs21resch11001@iith.ac.in}{0009-0000-2964-0368}{}

\author{Fahad Panolan 
}{School of Computer Science, University of Leads, UK}{F.Panolan@leeds.ac.uk}{ [0000-0001-6213-8687]}{}

\author{Geevarghese Philip 
}{Chennai Mathematical Institute, India}{gphilip@cmi.ac.in}{https://www.cmi.ac.in/~gphilip/}{}

\authorrunning{S. Aute, F. Panolan, and G. Philip} 

\Copyright{Shubhada Aute, Fahad Panolan, and Geevarghese Philip} 



\category{} 

\relatedversion{} 

\begin{document}

\maketitle

\begin{abstract}
Motivated by the landmark resolution of the 1-2-3 Conjecture, we initiate the study of the parameterized complexity of the Vertex-Coloring $\{0,1\}$-Edge-Weighting problem and its generalization, Vertex-Coloring Pre-edge-Weighting, under various structural parameters.
The base problem, Vertex-Coloring $\{0,1\}$-Edge-Weighting, asks whether we can assign a weight from $\{0,1\}$ to each edge of a graph. The goal is to ensure that for every pair of adjacent vertices, the sums of their incident edge weights are distinct. In the Vertex-Coloring Pre-edge-Weighting variant, we are given a graph where a subset of edges is already assigned fixed weights from $\{0,1\}$. The goal is to determine if this partial weighting can be extended to all remaining
edges such that the final, complete assignment satisfies the proper vertex coloring property.
While the existence of such weightings is well-understood for specific graph classes, their algorithmic complexity under structural parameterization has
remained unexplored.

We prove both hardness and tractability for the problem,
across a hierarchy of structural parameters. 
We show that both the base problem
and the Pre-edge-Weighting variant are W[1]-hard when parameterized by the size of a feedback vertex set (FVS) of the input graph.
On the positive side, we establish that the base problem and a restricted Pre-edge-Weighting variant where the pre-assigned weights are all $1$, become FPT when parameterized by the size of a vertex cover (VC) of the input graph. 
Further, we show that both the base problem and the Pre-edge-Weighting variant have  XP algorithms when parameterized by the treewidth of the input graph. 

These results constitute the first formal investigation into the parameterized complexity of the 
Vertex-Coloring $\{0,1\}$-Edge-Weighting problem and the Vertex-Coloring $\{0,1\}$-Pre-edge-Weighting, 
offering sharp boundaries between tractability and NP-hardness. They also add to the small list of problems that are known to be FPT when
parameterized by VC, but become W-hard for the parameter FVS.

\end{abstract}

\section{Introduction}\label{sec:introduction}
Let \(G=(V,E)\) be an undirected graph, and let \(S\) be a set of integers. A
weight function \(w:E(G) \to S\) is said to induce the \emph{color} \({\sf color}_{w}(v) =
\sum_{e \ni v} w(e)\)---the sum of all the weights of its incident edges---to
each vertex \(v\) of graph \(G\). Such a function \(w\) is said to be a
\emph{vertex-coloring edge-weighting} if this coloring is \emph{proper}---if
\({\sf color}_{w}(u) \neq {\sf color}_{w}(v)\) holds for every edge \(\{u,v\} \in E(G)\).
If there is a pair of adjacent vertices $u$ and $v$, such that ${\sf color}_w(u) = {\sf color}_w(v)$, then we call it a {\em coloring conflict}.

The famous \onetwothree in graph theory, posed by Karo\'nski, \L{}uczak, and
Thomason in 2004, proposed that every undirected simple graph without isolated
edges admits a vertex-coloring edge-weighting from the set \(\{1,2,3\}\). After
remaining open for nearly two decades and inspiring a lot of research that
yielded a number of related results, the conjecture was very recently resolved
in the affirmative by Keusch~\cite{keusch2024solution}.

In this work we take up the study of the parameterized complexity of the natural
\(\{0,1\}\)-edge-weighting variant of this problem, and of its pre-weighting
extension. We derive both FPT algorithms and W[1]-hardness results for these
problems under various parameterizations. 


All our graphs are finite, undirected, and simple, with no isolated edges. From
now on we consider weight functions of the form \(w:E(G) \to \{0,1\}\)---the
edge weights are all \(0\) or \(1\). For the sake of brevity we say that such a
weight function \(w\) is \emph{proper}---in place of ``vertex-coloring edge
weighting''---if it induces a proper coloring of the vertices. The primary
focus of this work is on the algorithmic problem of deciding the existence of a
proper weight function.

\defproblem{\vcew}%
{An undirected graph \(G=(V,E)\) on \(n\) vertices.}%
{If \(G\) admits a proper weight function, then output \yes. Otherwise, output
  \no.}

We also consider a generalization where some edges are pre-weighted.

\defproblem{\prvcew}%
{An undirected graph \(G=(V,E)\) on \(n\) vertices, with some edges pre-weighted
  by \(0\) or \(1\).}%
{If \(G\) admits a proper weight function that extends the pre-weighting, then
  output \yes. Otherwise, output \no.}

We establish parameterized tractability and intractability results for these
problems, under various natural parameters.

\smallskip
\noindent
{\bf Parameterized complexity results.}
We show that neither problem is likely to be \FPT for the parameter `size of a
smallest feedback vertex set'.

\begin{theorem}
  \label{thm:hardFVS} {\normalfont\vcew} and {\normalfont\prvcew} are
  \WoneHard parameterized by the size of a smallest feedback vertex set of
  the input graph.
\end{theorem}

We then look at the ``smaller'' parameter `treewidth of the input graph', for
which the W[1]-hardness directly carries over from \autoref{thm:hardFVS}. We show
that \vcew is in \XP for this parameter.

\begin{theorem}
  \label{thm:xpTreewidth} {\normalfont\vcew} can be solved in \XP time
  parameterized by the treewidth of the input graph: There is an algorithm which
  takes a graph \(G\) as input, runs in \XP time with respect to treewidth, and
  correctly decides if \(G\) admits a proper weight function. The algorithm
  outputs one proper coloring of the graph, if it exists.
\end{theorem}

The algorithm from Theorem~\ref{thm:xpTreewidth} can be easily modified to get an
\XP algorithm for the \prvcew problem. 
And it directly implies a polynomial-time algorithm for {\normalfont\vcew} on trees. Madarasi and Simon \cite{madarasi2022vertex} gave polynomial-time algorithm for trees for {\sc Vertex-Coloring} \{a,b\} {\sc Edge-Weighting}, where $a,b$ are any two integers. They also consider the {\sc Pre-Edge-Weighting} version of the problem for tress. 
Lyngsie~\cite{lyngsie2018neighbour} has
given a characterization of trees without a proper \(\{0,1\}\) edge-weighting.





Recall that the size of a smallest feedback vertex set of a graph lies between
the treewidth of the graph and the size of a smallest vertex cover of the graph.
In contrast to the hardness for the first two parameters, we show that
{\normalfont\vcew} becomes \FPT when the parameter is relaxed to be the size of
a smallest vertex cover of the input graph. In fact, we show tractability for a
restricted variant of the more general \prvcew problem as well, for this
parameter.

\begin{theorem}
  \label{thm:fptVC} {\normalfont\vcew} and {\normalfont\prvcew} where the
  pre-assigned weights are all 1, are both \FPT parameterized by the size of a
  smallest vertex cover of the input graph.
\end{theorem}


\noindent
{\bf Related Work.}
As we stated above, the \onetwothree proposed by Karo\'nski, \L{}uczak, and
Thomason~\cite{karonski2004edge} in 2004, posited that the edges of any
graph---with no isolated edges---can be weighted from \(\{1,2,3\}\) to induce a
proper vertex coloring via incident sums. This elegant problem remained open
for two decades, attracting significant attention and spawning numerous
variants~\cite{seamone2012conjecture}. Initial progress included verification
for 3-colorable graphs~\cite{karonski2004edge}, followed by finite upper bounds
on the least maximum positive integer edge weight required for guaranteeing a
proper vertex coloring. First, Addario-Berry et al.~\cite{addario2007thirty}
established that the numbers \(\{1, 2, 3, \dotsc, 30\}\) are enough. This was
later improved to \(16\)~\cite{addario2008sixteen} and then
\(13\)~\cite{wang2008thirteen}. A major breakthrough by Kalkowski, Karo\'nski,
and Pfender~\cite{kalkowski2010five} reduced this number to \(5\), and Thomassen
showed that \(4\) was sufficient for bipartite graphs~\cite{thomassen2010four}.
The conjecture was finally resolved in the affirmative by
Keusch~\cite{keusch2024solution} in 2024.

Efforts to resolve the \onetwothree led to the study of allied problems,
including vertex-coloring edge-weightings with restricted weight sets. A key
variant is the two-weight case, such as \(\{1,2\}\) or 
\(\{0,1\}\)%
~\cite{dudek2011complexity,lu2015vertex}. This problem is NP-hard in
general~\cite{dudek2011complexity,bensmail2013vertex}, even for edge weights
\(\{0,1\}\)~\cite{dudek2011complexity}. Positive results include existence for
certain bipartite graphs~\cite{lu2015vertex,chang2011vertex} and
characterizations of non-colorable bipartite graphs and
trees~\cite{lyngsie2018neighbour,khatirinejad2012vertex}. Some structural
insights for general $\{a,b\}$-weightings have also been worked
out~\cite{madarasi2022vertex,thomassen2016flow}.

%
There are problems motivated by \onetwothree, like minimizing the maximum color  \cite{bensmail2021minimizing}, and minimzing label sum \cite{bensmail2022proper} 
among all proper weight function using weights from $\{1,2, \dots , k\}$, where $k \geq 2$ are studied. 
It is proved that both these problems are NP-hard and polynomial time solvable on graphs of bounded treewidth~\cite{bensmail2022proper, bensmail2021minimizing}. In other words, these problems have XP algorithms parameterized by treewidth.  



\noindent
{\bf Organization of the rest of the paper.} In the next section we list
our notations and terminology.
We derive our
hardness results in \autoref{sec:hardness}. We describe our algorithms in
Sections \ref{sec:FPT} and \ref{sec:XP}. 
The algorithm for \prvcew when parameterized by the vertex cover number is in the section \ref{sec:prewt}.
We summarize our results and list some open problems
in \autoref{sec:conclusion}.


\section{Preliminaries and Notations}

For $\ell \in {\mathbb N}=\{1,2,\ldots\}$, we use $[\ell]$ to denote the set $\{1,2,\ldots,\ell\}$. The set of all integers $n$, such that $s \leq n \leq t$ is denoted as $[s,t]$. For a function $f:D \xrightarrow{} R$ and $A\subseteq D$, $f|_{A}$ denotes the function $f$ restricted on $A$. 

\smallskip

\noindent 
{\bf Graphs. }All graphs considered in this paper are simple undirected graphs without self-loops. Let $G$ be a graph with $V(G)$ as the set of vertices and $E(G)$ as the set of edges. We use $n=|V(G)|$ to denote the number of vertices. An edge between two vertices $u$ and $v$ is denoted by $\{u,v\}$ as well as $uv$. 
Let $d_G(v)$ denote the degree of a vertex $v$ in the graph $G$, and $\Delta(G)$ denote the maximum degree in $G$.
The neighborhood of a vertex $v$ in $G$ is the set of vertices adjacent to $v$, and we denote it as $N_G(v)$. When the context is clear, we drop the subscript $G$ and write only $N(v)$. 
For a vertex $v \in V(G)$, $G-v$ denotes the induced subgraph of $G$ with vertex set $V(G) \setminus \{v\}$ and all the edges incident on $v$ are deleted. 
Let $u,v \in V(G)$, and currently $uv \notin E(G)$. A graph obtained by adding an edge $uv$ to existing graph $G$ is denoted by $G+uv$. 
Let $C_n$ denote the cycle graph on $n$ vertices. 
%
Treewidth of a graph gives a measure to understand how close a graph is to a tree structure. 

\begin{definition}\cite{cygan2015parameterized}
A tree decomposition of a graph $G$ is a pair $\mathcal{T}=(T, \{B_t\}_{t \in V(G)})$ where $T$ is a rooted tree and every node $t$ in $T$ corresponds to a subset $B_t \subseteq V(G)$, called a bag,  such that the following conditions hold.  
\begin{itemize}
\item For every edge $e \in E(G)$, there exists a bag where both endpoints of $e$ belong to that bag. 
\item For every $v \in V(G)$, the set of nodes $t \in V(T)$ such that $v \in B_t$ induces a nonempty subtree $T_v$ of $T$. 
\item $V(G)=\bigcup_{t\in V(T)} B_t$
\end{itemize}
The width of a decomposition is $max_{t \in V(T)}|B_t|-1$. The treewidth of a graph $G$, denoted by ${\sf tw}(G)$, is the minimum width over all possible tree decompositions of $G$. 
\end{definition}
In this paper, we use the notion of {\em nice tree decomposition} to design a dynamic programming. 

\begin{definition}[\cite{cygan2015parameterized}]
A tree decomposition $\mathcal{T}=(T, \{B_t\}_{t \in V(G)})$ of a graph $G$ is a nice tree decomposition if it satisfies the following properties. 
The bags corresponds to the root node and the leaf nodes are empty. Every non-leaf node is one of the four types mentioned below.
\begin{itemize}
    \item \textbf{Introduce vertex}: Let $t$ be a node with only one child $t'$ such that $B_t = B_{t'} \cup \{v\}$, for some $v \notin B_{t'}$. We say that the vertex  $v$ is introduced in this node.
    \item \textbf{Introduce Edge}: Let $t$ be a node with one child $t'$, such that $B_t=B_{t'}$, and $uv$ is an edge in $G$, with $u,v \in B_t$. Edge $uv$ is introduced in this node. 
    \item \textbf{Forget node}: Let $t$ be a node with exactly one child $t'$ such that $B_t = B_{t'} \setminus \{w\}$. We say $w$ is forgotten in this node. 
    \item \textbf{Join node}: Let $t$ be a node with two children $t_1$ and $t_2$ such that $B_t = B_{t_1} = B_{t_2}$. 
\end{itemize}
\end{definition}

\begin{lemma}[\cite{cygan2015parameterized}]
    If a graph $G$ admits a tree decomposition of width $tw$,
then it also admits a nice tree decomposition of width  $tw$. Moreover,
given a tree decomposition $\mathcal{T} = (T, \{B_t\}_{t \in V (T )} )$  of $G$ of width $tw$, one can in time $O(tw^2 \cdot max(|V (T )|, |V (G)|))$ compute a nice tree decomposition of $G$ of width $tw$ that has at most $O(tw \cdot |V (G)|)$ nodes.
\end{lemma}


\noindent
{\bf Parameterized Complexity.}
We refer to \cite{cygan2015parameterized} for a detailed overview of parameterized algorithms and complexity. 
A parameterized problem is a language $L \subseteq \Sigma^* \times \mathbb{N}$ , where $\Sigma$ is a fixed, finite  alphabet. For an instance $(x,k) \in \Sigma^* \times \mathbb{N}$, $k$ is called the parameter. A parameterized problem $L \subseteq \Sigma^* \times \mathbb{N}$  is called fixed-
parameter tractable (FPT) if there exists an algorithm $\mathcal{A}$ (called a {\em fixed-parameter algorithm}), a computable function $f :\mathbb{N} \to \mathbb{N}$, and a constant $c$ such that, given $(x, k) \in \Sigma^* \times \mathbb{N}$, the algorithm $\mathcal{A}$  correctly decides whether $(x, k) \in L$ in time bounded by $f(k) \cdot |x|^c$. The complexity class containing all fixed-parameter tractable problems is called FPT.
A parameterized problem $L \subseteq \Sigma^* \times \mathbb{N}$ is called slice-wise polynomial (XP) if there exists an algorithm $\mathcal{A}$ and two  computable functions $f, g : \mathbb{N} \to \mathbb{N}$ such that, given $(x, k) \in \Sigma^* \times \mathbb{N}$, the algorithm $\mathcal{A}$ correctly decides whether $(x, k) \in L$ in time bounded by $f(k)\cdot|x|^{g(k)}$. The complexity class containing all slice-wise polynomial problems is called XP. 
A kernelization algorithm, or kernel, for a parameterized problem $Q$ is a polynomial time  algorithm $\mathcal{A}$ that, given an instance $(I,k)$ of $Q$, returns an equivalent instance $(I',k')$ of $Q$ such 
that $|I'|+k\leq g(k)$, 
for some computable function $g$. 
To prove a problem is \WoneHard, we need to give a {\em parameterized reduction} from a \WoneHard problem.

\begin{definition}[\cite{cygan2015parameterized}]
Let $A, B \subseteq \Sigma^* \times \mathbb{N}$ be two
parameterized problems. A {\em parameterized reduction} from $A$ to $B$ is an FPT algorithm that, given an instance $(x, k)$ of $A$, 
and 
outputs an instance $(x', k')$ of $B$ 
such that 
\begin{itemize}
    \item $(x, k)$ is a yes-instance of $A$ if and only if $(x', k')$ is a yes-instance of $B$, 
    \item $k' \leq g(k)$ for some computable function $g$.
\end{itemize}

\end{definition}
\section{FPT algorithm parameterized by vertex cover number}
\label{sec:FPT}
In this section, we give an FPT algorithm for \vcew parameterized by the vertex
cover number of the input graph. That is, the input consists of a graph \(G\)
and a vertex cover \(S\) of \(G\) of size \(k\). Let \(I = V(G) \setminus S\) be
the corresponding independent set. 
Note that a vertex in \(I\) has degree at most \(k\),
and hence the maximum value of a color it can get through any edge-weighting
\(w: E(G)\rightarrow \{0,1\}\) is at most \(k\). Interestingly, if \(G\) has a
\emph{proper} edge-weighting then there exists a proper edge-weighting of \(G\)
such that the color of any vertex in \(S\) is bounded by \(O(k^2)\).

\begin{lemma}
\label{lem:boundedcolor}
Let graph \(G\) be a yes-instance of \vcew. If \(G\) has vertex cover number
\(k\), then there exists a proper edge-weighting \(\hat{w}:E(G) \to \{0,1\}\)
such that for all vertices \(v\in V(G)\), \({\sf color}_{\hat{w}}(v) \leq
8k^2+8k\).
\end{lemma}

\begin{proof}
  Let \(S\) be a vertex cover of \(G\) with \(|S|=k\), and let \(I=V(G)\setminus
  S\) be the corresponding independent set. Set \(T=8k^2+8k\), and define the
  \emph{potential} of a proper edge-weighting \(w\) to be \(P(w) = \sum_{v \in
    V(G)}\max(0,\;{\sf color}_{w}(v)-T)\). The potential of \(w\) is thus the
  sum of the surpluses---if any---over \(T\), of the colors assigned by \(w\)
  to the vertices of \(G\). Let \(\hat{w}\) be a proper edge-weighting of \(G\)
  with the \emph{minimum potential} among all proper edge-weightings. We will
  show that \(P(\hat{w})=0\) holds, and this will imply the claimed bounds on
  \({\sf color}_{\hat{w}}(v)\).

  Assume for the sake of contradiction that \(P(\hat{w})>0\) holds. Then there
  is at least one vertex \(x\) for which \({\sf color}_{\hat{w}}(x)>T\) holds.
  Any vertex in the independent set \(I\) has degree at most \(k\)---it is
  adjacent only to vertices of \(S\)--- and so its color is at most \(k\). So
  the vertex \(x\) must be in the vertex cover \(S\). Let \(\bm{c} = {\sf
    color}_{\hat{w}}(x) > T\) be the color of \(x\) under \(\hat{w}\).

  Since \(x\) has at most \(k - 1\) neighbors in \(S\), the \(\hat{w}\)-weights
  of edges from \(x\) to other vertices in \(S\) can contribute at most \(k - 1\)
  to the color \(\bm{c}\). Let \(Y=\{y\in I : \hat{w}(xy) = 1\}\) be the set of all
  neighbors \(y\) of \(x\) in the independent set \(I\) for which \(\hat{w}\)
  assigns the weight \(1\) to the edge \(xy\). Then we have
\begin{equation}
    |Y| \geq \bm{c} - (k - 1) \geq T + 1 - (k - 1) = 8k^{2} + 7k + 2.
    \label{eqn:y-lower-bound}
  \end{equation}

  Let \(C\) be the set of all the colors that the edge-weighting \(\hat{w}\)
  gives to the vertices of \(G\). As noted above, each vertex in the independent
  set \(I\) has a color in the range \(\{0,1,\dots,k\}\). And since \(|S| =
  k\), \(\hat{w}\) can assign at most \(k\) different colors to the vertices in
  \(S\). So we get that \(|C| \leq (k + 1) + k = 2k + 1\) holds. Let \(D =
  \{d\in\{1,2,\dotsc,7k + 1\}\mid \bm{c} - d \notin C\}\) be the set of all
  integers \(d\) in the range \(\{1,2,\dotsc,7k + 1\}\) such that the integer
  \(\bm{c} - d\) is \emph{not} one of the colors assigned by the edge-weighting
  \(\hat{w}\) to some vertex of \(G\).

  Since (by construction) \(\bm{c} > 8k^{2} + 8k\) holds, the inequality
  \(\bm{c} > 7k + 1\) holds for all positive integers \(k\), and so all the
  values in the set \(\{\bm{c} - d \mid d \in \{1,2,\dotsc,7k + 1\}\}\) are
  positive integers. Since the values \(\bm{c} - d\) are distinct for different
  integers \(d\), we get that \(|D| \geq (7k + 1) - |C| \geq 7k + 1 - (2k + 1) =
  5k\) holds. Let \(d_{1}\) be the smallest element of the set \(D\). Since
  \(|C| \leq 2k + 1\) holds, we get that \(d_{1} \leq 2k + 2\) holds.

  Consider the following greedy procedure to reduce the color of vertex \(x\),
  by starting with the edge-weighting \(\hat{w}\) and changing the weights from
  \(1\) to \(0\) for some edges \(xy\) for vertices \(y \in Y\): Start by
  setting \(q = 0\), and then go over the vertices \(y \in Y\) one by one in
  some arbitrary order. When we consider a vertex \(y\) the color of vertex
  \(x\) is \(\bm{c} - q\). If setting the weight of edge \(xy\) to \(0\) would
  create a color conflict for vertex \(y\) with any of its neighbors (including
  vertex \(x\)), we skip \(y\); otherwise we change the weight of edge \(xy\) to
  \(0\) and increment \(q\) by \(1\). We stop when \(q\) becomes \(d_{1}\) or
  when all vertices of \(Y\) have been examined, whichever happens first. Let
  \(\tilde{w}\) denote the set of edge weights when this procedure stops.

  \begin{claim}
    \(q < d_{1}\) must hold at the end of this procedure.
    \label{claim:q-not-d1}
  \end{claim}
  \begin{proof}
    Suppose the procedure sets \(q = d_{1}\). Then (i) \({\sf color}_{\tilde{w}}(x) =
    \bm{c} - d_{1} \notin C\), (ii) for each edge \(xy\) whose weight was
    changed to \(0\) the colors of the vertices \(x\) and \(y\) decreased by
    \(1\) each, and (iii) all other vertex colors are unchanged.

    Since \({\sf color}_{\tilde{w}}(x) \notin C\), no neighbor of \(x\) in the
    vertex cover \(S\) gets the same color as \(x\) under the new edge-weighting
    \(\tilde{w}\). Since \({\sf color}_{\tilde{w}}(x) = \bm{c} - d_{1} \geq
    8k^{2} + 8k + 1 - (2k + 2) = 8k^{2} + 6k - 1 > k\) holds for all \(k \geq
    1\), no neighbor of \(x\) in the independent set \(I\) gets the same color
    as \(x\) under \(\tilde{w}\). Thus no neighbor of \(x\) in \(G\) gets the
    same color as \(x\) under \(\tilde{w}\).

    Consider a vertex \(y\in I\) for which the procedure changed the weight of
    edge \(xy\) to \(0\). This change was done only because at that point in the
    procedure it caused no color conflicts for \(y\) with any of its neighbors
    (which are all in the set \(S\)). Vertex \(x\) is only neighbor of \(y\)
    whose color may change later during the procedure, and from the previous
    paragraph we get that \({\sf color}_{\tilde{w}}(x) \neq {\sf
      color}_{\tilde{w}}(y)\) holds when the procedure ends. Thus the new
    edge-weighting \(\tilde{w}\) results in no color conflicts involving any
    vertex in the independent set \(I\).

    Thus we get that the new edge-weighting \(\tilde{w}\) is proper. And when going
    from \(\hat{w}\) to \(\tilde{w}\) the color of \(x\) dropped by \(d_{1}\geq 1\),
    while all other vertex colors either decreased or remained unchanged. Thus
    \(P(\tilde{w}) < P(\hat{w})\) holds, which contradicts the minimality of
    \(\hat{w}\).
  \end{proof}
  Let \(Y_{\text{dropped}} \subseteq Y\) denote the subset of vertices \(y\in
  Y\) such that the weight of the edge \(xy\) was changed from \(1\) to \(0\) by
  the above procedure. Since \(d_{1} \leq 2k + 2\) holds we get from
  \autoref{claim:q-not-d1} that \(|Y_{\text{dropped}}| = q \leq 2k + 1\) holds.
  Set \(Y' = Y\setminus Y_{\text{dropped}}\). Then from
  \autoref{eqn:y-lower-bound} we get
  \begin{equation}
    |Y'| = |Y| - |Y_{\text{dropped}}| \geq 8k^2 + 7k + 2 - 2k - 1 = 8k^2 + 5k + 1.
    \label{eqn:yprime-lower-bound}
  \end{equation}

  Pick an arbitrary vertex \(y \in Y'\). During the greedy procedure we considered
  dropping the weight of edge \(xy\) to \(0\), and found that this would create a
  color conflict between \(y\) and some neighbor of \(y\). At that point the
  color of \(x\) was at least \(\bm{c} - q \geq T + 1 - (d_{1} - 1) = T - d_{1} +
  2 \geq 8k^{2} + 8k - 2k - 2 + 2 = 8k^{2} + 6k > k + 2\) (for \(k \geq 1\)). So
  the conflicting neighbor of \(y\) could not have been \(x\). And since the
  procedure does not change the \(\hat{w}\)-color of this vertex \(y\)
  or of any vertex other than \(x\) in the vertex cover \(S\), we get that there
  exists a vertex \(z_{y} \in S\setminus\{x\}\) such that (i) \({\sf
    color}_{\hat{w}}(z_{y}) = {\sf color}_{\hat{w}}(y) - 1\) and (ii) \(z_{y}y\in
  E(G)\) both hold. Since \({\sf color}_{\hat{w}}(y)\leq k\)---because vertex
  \(y\) belongs to the independent set \(I\)---the first condition implies that
  \({\sf color}_{\hat{w}}(z_{y})\leq k - 1\) holds for this vertex \(z_{y} \in
  S\).

  Thus for every vertex \(\tilde{y} \in Y'\) there exists at least one vertex
  \(\tilde{z} \in S\setminus\{x\}\) such that (i) \({\sf
    color}_{\hat{w}}(\tilde{z}) = {\sf color}_{\hat{w}}(\tilde{y}) - 1\) and (ii)
  \(\tilde{z}\tilde{y} \in E(G)\) both hold. Now since \(|S\setminus\{x\}| \leq
  k-1\), we get from \autoref{eqn:yprime-lower-bound} and an averaging argument
  that there is a vertex \(z\in S\setminus\{x\}\) and a set \(Y''\subseteq Y'\)
  with \(|Y''| \geq \frac{8k^{2} + 5k + 1}{k - 1} > 8k + 13\) such that for
  \emph{every} \(y \in Y''\) both (i) \({\sf color}_{\hat{w}}(z) = {\sf
    color}_{\hat{w}}(y) - 1\) and (ii) \(zy \in E(G)\) hold. Let \(t = {\sf
    color}_{\hat{w}}(z)\); then \(t \leq k - 1\), and \({\sf color}_{\hat{w}}(y) =
  t + 1\) holds for all \(y \in Y''\). See Figure~\ref{vc} for an illustration. Yellow rectangle in the figure represents $Y''$.

\begin{figure}[h]
    \centering
 \includegraphics[width=0.7\linewidth]{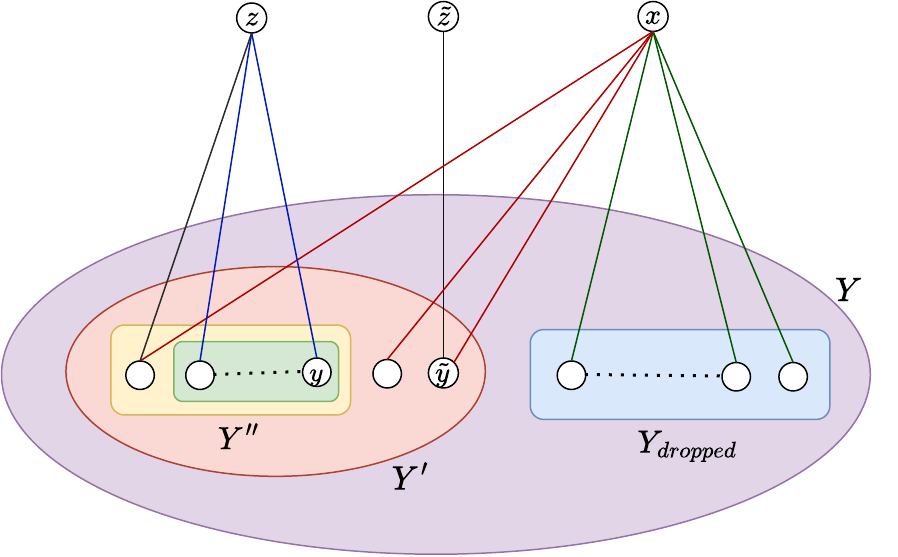}
    \caption{Illustration of proof of Lemma~\ref{lem:boundedcolor}}
    \label{vc}
\end{figure}

  Since \({\sf color}_{\hat{w}}(z) = t\), at most \(t \leq k - 1\) edges incident
  with the vertex \(z\) can have weight \(1\) in \(\hat{w}\). Hence among the
  edges \(zy\) with \(y \in Y''\), at least \(|Y''| - (k - 1) \geq 8k + 14 - (k -
  1) = 7k + 15\) have weight \(0\) in \(\hat{w}\). Let \(Y_{0} \subseteq Y''\) be
  a set of vertices where \(\hat{w}(zy) = 0\) for all \(y \in Y_{0}\), and with
  \(|Y_{0}| \geq 7k + 15\).

  Recall that \(|D| \geq 5k\), and that for all \(d \in D\), \(d \leq 7k + 1\)
  and \(\bm{c} - d \notin C\) both hold. Now since \(|C| \leq 2k + 1\) holds, at
  most \(2k + 1\) elements \(d \in D\) satisfy the condition \(t + d \in C\). So
  there are at least \(|D| - (2k + 1) \geq 3k - 1\) elements \(d\in D\) that
  satisfy the complementary condition \(t + d \notin C\). Choose the
  \emph{largest} element \(\hat{d} \in D\) which satisfies this latter
  condition. Observe that all of the following hold for this \(\hat{d}\) :
  \(\hat{d} \leq 7k + 1 \leq |Y_{0}|\), \(\hat{d} \geq (7k + 1) - (2k + 1) =
  5k\), \(\bm{c} - \hat{d} \notin C\), and \(t + \hat{d} \notin C\).

  Select an arbitrary subset \(Y_{1}\subseteq Y_{0}\) with \(|Y_{1}| = \hat{d}\). It is represented with green rectangle in Figure \ref{vc}.
  Starting with the original proper edge-weighting \(\hat{w}\), define a new
  edge-weighting \(w'\) as follows:
  \begin{itemize}
  \item For every \(y\in Y_{1}\), set \(w'(xy) = 0\) (was \(1\)) and \(w'(zy) =
    1\) (was \(0\)).
  \item Leave all other edge-weights unchanged; for all other edges \(e\), set
    \(w'(e) = \hat{w}(e)\).
  \end{itemize}

  \begin{claim}
 \(w'\) is a proper edge-weighting of \(G\).
    \label{claim:proper-edge-weighting}
  \end{claim}
\begin{claimproof}

    \begin{itemize}
    \item Vertex \(x \in S\): new color \(\bm{c} - \hat{d} \geq 8k^{2} + 8k + 1 -7k - 1 =
      8k^{2} + k\). The neighbors of \(x\) are in \(S\cup I\).
      \begin{itemize}
      \item New colors of the vertices in \(S\setminus \{x, z\}\) belong to the
        set \(C\), and since \(\bm{c} - \hat{d} \notin C\) there is no conflict
        between \(x\) and any of its neighbors in \(S\setminus \{z\}\). As for
        the vertex \(z\): its new color is \(t + \hat{d} \leq k - 1 + 7k + 1 =
        8k\) which is strictly smaller than the lower bound \(8k^{2} + k\)
        (derived above) for the new color of \(x\) for all \(k \geq 1\). So
        there is no conflict possible between \(x\) and \(z\).
      \item New colors of the vertices in \(I\) are at most \(k\) each, which
        is strictly smaller than the lower bound \(8k^{2} + k\) for the new
        color of \(x\) for all \(k \geq 1\). So there is no conflict possible
        between \(x\) and any vertex in \(I\).
      \end{itemize}
    \item Vertex \(z \in S\): new color \(t + \hat{d}\) satisfies \(5k \leq t +
      \hat{d} \leq 8k\). The neighbors of \(z\) are in \(S\cup I\).
      \begin{itemize}
      \item New colors of the vertices in \(S\setminus \{x, z\}\) belong to the
        set \(C\), and since \(t + \hat{d} \notin C\) there is no conflict
        between \(z\) and any of its neighbors in \(S\setminus \{x, z\}\). As
        for the vertex \(x\): its new color is \(\geq 8k^{2} + k\) which is
        strictly larger than the upper bound \(8k\) for the new color of \(z\),
        for all \(k \geq 1\). So there is no conflict possible between \(x\) and
        \(z\).
      \item New colors of the vertices in \(I\) are at most \(k\) each, which
        is strictly smaller than the lower bound \(5k\) for the new
        color of \(z\) for all \(k \geq 1\). So there is no conflict possible
        between \(z\) and any vertex in \(I\).
      \end{itemize}

    \item Vertices \(y\in Y_{1}\): new color \(t + 1 \leq k\) is the same as
      its \(\hat{w}\)-color. Its neighbors are all in the vertex cover \(S\).
      Of these, the new colors of the vertices in \(S\setminus \{x, z\}\) are
      the same as their \(\hat{w}\)-colors, and since \(\hat{w}\) was a proper
      edge-weighting, there are no conflicts possible between any such \(y\) and
      the vertices in \(S\setminus \{x, z\}\) under the new edge-weighting
      \(w'\), as well. The new color of vertex \(x\) is at least \(8k^{2} +
      k\), and the new color of vertex \(z\) is at least \(5k\). Both these
      numbers are larger than the upper bound \(k\) on the color of any vertex
      \(y \in Y\), and so there is no conflict possible between such \(y\) and
      any vertex in \(S\).
    \item All other vertices have the same colors as before, so no conflicts
      are possible involving them.\qedhere
    \end{itemize}

\end{claimproof}

  In going from \(\hat{w}\) to \(w'\) the color of vertex \(x\) decreased by
  \(\hat{d}\geq 5k > 5\) (for any \(k \geq 1\)), so its potential strictly
  drops. The color of \(z\) increased to \(t + \hat{d} \leq k - 1 + 7k + 1 = 8k
  < T\) (for any \(k \geq 1\)), so its potential remains \(0\). All other
  vertices have potentials unchanged. Hence \(P(w') < P(\hat{w})\),
  contradicting the minimality of \(\hat{w}\). So our assumption \(P(\hat{w}) >
  0\) must be false; we must have \(P(\hat{w}) = 0\). That is, under the
  edge-weighting \(\hat{w}\) every vertex has color at most \(T = 8k^{2} +
  8k\). 
\end{proof}

\begin{lemma}
\label{lem:kernelVC}
\vcew parameterized by the vertex cover number $k$, admits a kernel $(H,k')$ of size $O(4^k k^4)$ where $k'\leq k$. 
\end{lemma}

\begin{proof}
    Let $(G,k)$ be the input instance. That is, the vertex cover number of $G$ is $k$.  We compute a maximal matching of the graph, say $M$. It can be computed using a greedy algorithm in polynomial time in $O(nk)$. Let $S$ denote the set of endpoints of edges in $M$. As vertex cover number is $k$, $|S| \leq 2k$. For every vertex in $I = V(G) \setminus S$, we define a relation as follows. Vertex $u$ is related to vertex $v$ if $N(u)=N(v)$. This is an equivalence relation, and partitions $I$ into at most $2^{2k}$ equivalence classes. For a subset $S' \subseteq S$, let $A_{S'}$ denote an equivalence class of $I$ corresponding to $S'$. That is, for every vertex $u \in A_{S'}$, $N(u) =S'$, and $|S'| \leq 2k$.
As $I$ is very large, there may be many equivalence classes of large size. First, we show that we can consider only a bounded number of vertices from each equivalence class.

\begin{claim}
\label{claimvcred}
    If there is a large equivalence class $A_{S'}$, with $|A_{S'}| > 2k(8k^2+8k)$, then all but  $2k(8k^2+8k)$ vertices will get color $0$ by any proper weight function (if it exists) that satisfies Lemma~\ref{lem:boundedcolor}.
\end{claim} 
\begin{claimproof}
For contradiction, assume there is a proper weight function $w$ that satisfies Lemma \ref{lem:boundedcolor}, and the claim is not true. That is, for each vertex $v\in V(G)$, ${\sf color}_w(v)\leq 8k^2+8k$. 
Then, there will be more than $2k(8k^2+8k)$ vertices with color at least $1$, implying that there is at least one edge incident on each of them of weight $1$. Note that their neighbors are only in $S'$, with $|S'| \leq 2k$. Then by pigeonhole principle, there will be at least one vertex $x \in S'$, with more than $8k^2+8k$ edges of weight $1$ incident on it from $A_{S'}$. Hence, ${\sf color}_w(x) > 8k^2+8k$, a contradiction to the fact that $w$ satisfies Lemma~\ref{lem:boundedcolor}. 
\end{claimproof}
The above claim leads to the following 
reduction rule. 

\smallskip
\noindent
\textbf{Reduction Rule}: For every large equivalence class $A$, with $|A| > 2k(8k^2+8k)+1$, delete all but $2k(8k^2+8k)+1$ vertices from $A$.

Let $H$ be the graph obtained after applying the above reduction rule and $H$ is the output of our kernelization algorithm. Next, we prove that $G$ has a proper weight function if and only if $H$ has a proper weight function. Since for any two vertices $u_1$ and $u_2$ in the same equivalence class, $G-u_1$ is isomorphic to $G-u_2$, by Claim~\ref{claimvcred}, if $G$ has a proper weight function, then $H$ has a proper weight function. 

Conversely, suppose $H$ has a proper weight function $w$ that satisfies Lemma~\ref{lem:boundedcolor}. Then, we define a weight function $w'$ on $E(G)$ as follows. For all $e\in E(H)$, $w'(e)=w(e)$ and for all $e\in E(G)\setminus E(H)$, $w(e)=0$. Observe that for all $v\in V(H)$, ${\sf color}_{w'}(v)={\sf color}_{w}(v)$ and for all $u\in V(G)\setminus V(H)$, ${\sf color}_{w'}(u)=0$. Next, we prove that $w'$ is proper. Let $e=uv$ be an edge in $E(G)$. We have two cases. The first case is $e\in E(H)$. Then, we know that $u,v\in V(H)$ and ${\sf color}_{w'}(v)={\sf color}_{w}(v)$ and ${\sf color}_{w'}(u)={\sf color}_{w}(u)$. Since $w$ is a proper weight function, 
${\sf color}_{w}(u)\neq {\sf color}_{w}(v)$. Thus, we get that ${\sf color}_{w'}(v)\neq {\sf color}_{w'}(u)$
In the second case $e=uv\in E(G)\setminus E(H)$. Then, one endpoint of $e$ is in $V(H)$ and the other is in $V(G)\setminus V(H)$. Let $v\in V(H)$ and $u\in V(G)\setminus V(H)$. We know that ${\sf color}_{w'}(u)=0$ and ${\sf color}_{w'}(v)= {\sf color}_{w}(v)$. To prove 
${\sf color}_{w'}(v)\neq {\sf color}_{w'}(u)$, it is enough to prove that ${\sf color}_{w'}(v)\neq 0$. Let $u$ belongs  to the equivalence class $A_{S'}$ where $S'=N_G(u)$. Clearly $v\in S'$ and there exactly $2k(8k^2+8k)+1$ vertices $x$ in $V(H)$ such that $N_H(x)=S'$. Since $w$ satisfies Lemma~\ref{lem:boundedcolor}, 
there is a vertex in $x\in A\cap V(H)$, such that that ${\sf color}_{w}(x)=0$. 
This implies that ${\sf color}_{w}(v)\neq 0$ because $xv\in E(H)$.


Now we bound the size of $H$. There are at most $2^{2k}$ equivalence classes, each of size at most $2k(8k^2+8k)+1$. Hence, the number of vertices in $H$ is at most $2k+2^{2k} \cdot (2k(8k^2+8k)+1)=O(4^k k^3)$. Since the vertex cover number of $H$ is at most $k$, the number of edged in $H$ is $O(4^k k^4)$. 

\end{proof}

\begin{theorem}
There is an FPT algorithm of running time 
$2^{O(k^4)}n$ for \vcew parameterized by the vertex cover number $k$.
\end{theorem}

\begin{proof}
Let $G$ be the input graph. First we apply the kernelization algorithm and get a kernel $H$ of size $O(4^k k^4)$. We also know that the vertex cover number of $H$ is at most $k$. If $H$ is a yes-instance, then by Lemma~\ref{lem:boundedcolor}, there is a weight function $w$ such that 
the color induced on every vertex is at most $8k^2+8k$. This implies that the number of edges assigned one by $w$ is at most $k(8k^2+8k)$. So we choose all edge subsets of $E(H)$ of size at most $k(8k^2+8k)$ and test in linear time if setting those edges one and the remaining zero is a proper weight function or not. 
The number of edge subsets of size at most $k(8k^2+8k)$ in $E(H)$ is at most $2^{O(k^4)}$. Thus, the total running time follows.   
\end{proof}

\section{XP algorithm parameterized by treewidth}
\label{sec:XP}
In this section, we provide an XP algorithm for solving the \vcew  with treewidth as the parameter. 
If a graph $G$ is a yes-instance of \vcew, then there is a proper weight function $w$ that induces a proper coloring. There can be several such proper weight functions. 
For a proper weight function $w$, we define a spanning subgraph $H^w$ of $G$ as follows. The vertex set of $H^w$ is $V(H^w)=V(G)$ and the edge set is $E(H^w)=\{e\in E(G)~:~w(e)=1\}$. 
If a vertex $v$ has $d$ of its incident edges assigned weight $1$ by $w$, then its color is $d$, and only these $d$ edges, will be in the subgraph $H^w$. Hence, ${\sf color}_w(v) = d_{H^w}(v) =d$. As $w$ is a proper weight function, any two adjacent vertices in $G$ have different colors. Equivalently, any two adjacent vertices of $G$ have distinct degrees in $H^w$.
We call such a spanning subgraph $H$ as a {\em solution} to $G$. 
A proper weight function and its corresponding spanning sugraph $H^w$ are equivalent. In the sense that if $w$ is given, then there is a unique spanning subgraph $H^w$ that is associated with $w$. And if some spanning subgraph $H$ is given that satisfies the condition that any two adjacent vertices in $G$ have distinct degrees in $H$, we can uniquely construct the corresponding proper weight function.
We develop an algorithm that provides a spanning subgraph $H$, such that degree of adjacent vertices is distinct, if it exists. We can get a proper weight function $w$ by setting weight $1$ to all the edges in $H$, and $0$ otherwise.

So, our algorithm outputs a subgraph $H^w$ if one exists.
We use dynamic programming to construct such a subgraph in a bottom-up approach. 
We utilize a nice tree decomposition 
for this purpose. 
Let $G$ be the input graph and treewidth of $G$ is $tw$. We assume that a tree decomposition
tree decomposition $\mathcal{T} = (T, \{B_t\}_{t \in V (T )} )$  of $G$, of width $tw$ is given as part of the input. Because there is an algorithm that given a graph $G$ and integer $tw$, runs in time $2^{O(tw^2)}n^{O(1)}$ and outputs a tree decomposition of $G$ of width at most $tw$, if it exists~\cite{TWDT}.
For a node $t\in V(T)$, 
let $G_{t}$ be the subgraph of $G$ on the vertices that appear in the bags of the subtree of $T$ rooted at $t$.

In dynamic programming, to reach a solution graph $H$, we store, at each node, a set of partial solutions. Clearly, a partial solution at a node $t$ is a spanning subgraph $H_t$ of $G_t$ that satisfy some properties. Before explaining those properties, we need define a {\em state} at a node $t\in V(T)$. Towards that let $H$ is a solution to $G$. Let $E_t=E(G_t)\cap E(H)$. Consider the graph $H_t=(V(G_t),E_t)$. Clearly, we want $H_t$ to be a partial solution at $t$.   
%
%
%
%
%
%
For a vertex $v$ in the bag $B_{t}$, we store two numbers $FD(v)$, the (final) degree of $v$ in $H$. And $CD_t(v)$, the (current) degree in $H_t$. 
A pair of final degree and current degree is {\em feasible} if $CD(v) \leq FD(v)$. 
Let $P_v$ denote the set of all feasible pairs for a vertex $v$. That is, $P_v = \{(0,0),(1,0),(2,0), (2,1), (2,2), \dots, (d_G(v), d_G(v))\}$. 
Let $\Delta$ be the maximum degree in $G$. 
Now we define a {\em state} at a node $t$. 

\begin{definition}
 A {\em state} at a node $t$ is a function  $f ~:~ B_t \xrightarrow{}\{0,1,\ldots, \Delta\} \times \{0,1,\ldots, \Delta\}$ such that for all $v\in B_t$, $f(v)\in P_v$.   
\end{definition}

\begin{definition}
\label{def:part}
Let $f$ be a state at a node $t\in V(T)$. A spanning subgraph $H_t$ of $G_t$ is a {\em partial solution of state $f$} at $t$ if the following holds.   
\begin{enumerate}
    \item For each $v\in B_t$, $f(v)=(x,d_{H_t}(v))$ for some $x\in \{0,1,\ldots, \Delta\}$. 
    \item For each edge $uv\in E(G_t)$ the following holds. 
    \begin{itemize}
        \item[$(i)$] If $u,v\notin B_t$, then  $d_{H_t}(u)\neq d_{H_t}(v)$. 
        \item[$(ii)$] If $u\in B_t$ and $v\notin B_t$, then $d_{H_t}(v)\neq x$, where $(x,d_{H_t}(u))=f(u)$. 
        \item[$(iii)$] If $u,v\in B_t$, then $x_1\neq x_2$, where $(x_1,d_{H_t}(u))=f(u)$ and $(x_2,d_{H_t}(v))=f(v)$.  
    \end{itemize}
\end{enumerate}
\end{definition}

\begin{definition}
Let $f$ be a state at a node $t\in V(T)$ 
and $H_t$ is a partial solution of state $t$. We say that $(H_t,f)$ 
is {\em extendable} if there is a solution $H$ to $G$ such that the following holds. 
\begin{itemize}
    \item $H_t$ is a subgraph of $H$ and $E(H)\cap E(G_t)=E(H_t)$. 
    \item For all $v\in B_t, f(v)=(d_{H}(v),d_{H_t}(v))$. 
\end{itemize}
\end{definition}

\begin{lemma}
\label{lem:PSext}
Let $f$ be a state at a node $t\in V(T)$. Let $H_1$ and $H_2$ be two partial solutions of state $f$ at $t$. If $(H_1,f)$ is extendable, then $(H_2,f)$ is extendable.   
\end{lemma}
\begin{proof}
    Since $(H_1,f)$ is extendable, there is a solution $H$ to $G$ such that following holds. 
\begin{itemize}
    \item[$(i)$] $H_1$ is a subgraph of $H$ and $E(H)\cap E(G_t)=E(H_1)$. 
    \item[$(ii)$] For all $v\in B_t, f(v)=(d_{H}(v),d_{H_1}(v))$. 
\end{itemize}
Now we define a spanning subgraph $H'$ of $G$ as follows. The vertex set of $H'$ is $V(G)$ and the edge set is $E(H_2)\cup (E(H)\setminus E(H_1))$. 
From $(i)$, we get that $E(H')\setminus E(G_t)=E(H)\setminus E(G_t)$. 
This implies that $E(H')\cap E(G_t)=E(H_2)$ and $H_2$ is a subgraph of $H'$. From $(ii)$ and the fact that both $H_1$ and $H_2$ are of state $f$, we get that for any vertex $v\in B_t$, $f(v)= (d_{H}(v),d_{H_1}(v))=(d_{H}(v),d_{H_2}(v))$. 

Finally, we prove that $H'$ is indeed a solution to $G$.  
That is, we want to prove that for any edge $uv \in E(G)$, $d_{H'}(u) \neq d_{H'}(v)$. First, we prove that  $(a)$ for any vertex $x\in B_t$, $d_{H}(x)=d_{H'}(x)$. Let $x\in B_t$. Since both $H_1$ and $H_2$ are partial solutions of state $f$ at $t$, $d_{H_1}(x) = d_{H_2}(x)$. Also, $E(H')\setminus E(G_t)=E(H)\setminus E(G_t)$. 
With these two equations, we see that $d_H(x) = d_{H'}(x)$. Now fix an edge $uv \in E(G)$. Now we have the following cases. 
\begin{itemize}
    \item Suppose $u,v \in B_t$. Then $d_H(u) = d_{H'}(u)$ and $d_H(v) = d_{H'}(v)$, by statement $(a)$ above. As $H$ is a solution to $G$, so $d_H(u) \neq d_{H}(v)$. Hence, $d_{H'}(u) \neq d_{H'}(v)$. 
\item If $u,v \in V(G) \setminus V(G_t)$, then because of $E(H')\setminus E(G_t)=E(H)\setminus E(G_t)$, $d_H(x) = d_{H'}(x)$, for $x \in \{u,v\}$. But $H$ is a solution, so $d_H(u) \neq d_{H}(v)$. Hence, $d_{H'}(u) \neq d_{H'}(v)$.
\item If $u,v \in V(G_t) \setminus B_t$, then $d_{H_2}(x) = d_{H'}(x)$, for $x \in \{u,v\}$. As $H_2$ is a partial solution, $d_{H_2}(u) \neq d_{H_2}(v)$. And so $d_{H'}(u) \neq d_{H'}(v)$.
\item If $u \in B_t$, and $v \in V(G) \setminus V(G_t)$, 
then $d_H(u) = d_{H'}(u)$ (because of $(a)$) and $d_H(v) = d_{H'}(v)$ (because $E(H')\setminus E(G_t)=E(H)\setminus E(G_t)$). Thus, since $H$ is a solution, we have $d_H(u) \neq d_H(v)$ and this implies $d_{H'}(u) \neq d_{H'}(v)$. 
\item Suppose $u \in B_t$ and $v \in V(G_t) \setminus B_t$. By statement $(a)$, we have $d_H(u)=d_{H'}(u)$. By statement $(ii)$ above, we have $f(u)=(d_H(u),d_{H_1}(u))$. Also, since $H_2$ is a partial solution of state $f$ at $t$, we get $d_{H_2}(v)\neq d_{H}(u)$. This follows from condition $2(ii)$ of Definition~\ref{def:part}. This implies that $d_{H_2}(v)\neq d_{H'}(u)$, 
because $d_H(u)=d_{H'}(u)$. Since 
$d_{H_2}(v)=d_{H'}(v)$, we get $d_{H'}(v)\neq d_{H'}(u)$. \
\end{itemize}
\end{proof}

Because of Lemma~\ref{lem:PSext}, for any node $t\in V(T)$ and state $f$ at $t$, we only need to store one partial solution of state $f$ (if it exists). 
For convenience, we say that for a node $t$ with $B_t=\emptyset$, there is exactly one state which is the empty function $f_{\emptyset}~:~ \emptyset \xrightarrow{} \{0,1,\ldots,\Delta\}\times \{0,1,\ldots,\Delta\}$.    
Clearly, the output is the partial solution of state $f_{\emptyset}$ stored at the root node (if it exists).

Next we explain how to compute those partial solution. 
Our dynamic programming(DP) algorithm will create a DP table 
${\cal D}$. This table is indexed with nodes $t$ in $T$ and states at node $t$. 
That is, for every node $t$ and a state $f$ at $t$, we compute and store a partial solution $H_f$ of state $f$ at $t$ in the table entry ${\cal D}[t,f]$, if it exists. Otherwise we store $\bot$ at ${\cal D}[t,f]$. As mentioned earlier we do a dynamic programming over $T$ in a bottom up fashion. Let $t$ be a node in $T$ such that we have computed 
${\cal D}[t',f']$ for all descendent  $t'$ of $t$ and 
state $f'$ at node $t'$
Now we explain how to compute partial solutions at node $t$.

\begin{itemize}
    \item \textbf{Case 1: $t$ is a leaf node.} There is only one state at node $t$ which is the empty function $f_{\emptyset}$. We store ${\cal D}[t,f_{\emptyset}]:=(\emptyset,\emptyset)$. Here, $(\emptyset,\emptyset)$ is the empty graph.  
    \item \textbf{Case 2: $t$ is an introduce vertex node.} Let $v$ be the vertex that is introduced in $t$ and $t'$ be the child of $t$. We know that $d_{G_t}(v)=0$. Let $f$ be an arbitrary state at $t$. Let $f(v)=(x,y)$. If $y>0$, then we store ${\cal D}[t,f]:=\bot$. Now consider the case $y=0$. Let $f_1,\ldots,f_{\ell}$ be the states at $t'$ such that $f_i=f|_{B_{t'}}$ for all $i\in [\ell]$.  
    If there exists $i\in [\ell]$ such that 
    ${\cal D}[t',f_i]\neq \bot$, 
    then we arbitrarily choose one such  $i$ and store ${\cal D}[t,f]:= {\cal D}[t',f_i]$.
    Otherwise, we store ${\cal D}[t,f]:=\bot$. 
    
    \item \textbf{Case 3: $t$ is an introduce edge node.} Let $e=uw$ be the edge introduced in $t$ and $t'$ be the child of $t$. Let $f$ be a state at $t$. 
    Observe that $f$ is a state at $t'$ as well. 
    Let $f(u)=(x_1,y_1)$ and $f(w)=(x_2,y_2)$. 
    \begin{enumerate}
        \item If $x_1=x_2$, then  ${\cal D}[t,f]:=\bot$.  
        \item  If $x_1\neq x_2$ and ${\cal D}[t',f]\neq \bot$, then ${\cal D}[t,f]:= {\cal D}[t',f]$. 
        \item Let $f'$ be the function defined as follows. For all $z\in B_{t}\setminus \{u,w\}$, $f'(z)=f(z)$, $f'(u)=(x_1,y_1-1)$ and $f'(w)=(x_2,y_2-1)$. If $x_1\neq x_2$ and $y_1-1,y_2-1\geq 0$, and 
        ${\cal D}[t',f']=H_{f'}\neq \bot$, then ${\cal D}[t,f]:= H_{f'}+e$.
        \item If none of the above is applicable, then ${\cal D}[t,f]:=\bot$.
\end{enumerate}
    
    \item \textbf{Case 4: $t$ is a forget vertex node.} Let $v$ be the vertex forgotten at $t$ and let $t'$ be the child of $t$. Notice that $B_{t'}=B_t\cup \{v\}$. 
Let $f$ be an arbitrary state at $t$. Let $f_1,\ldots,f_{\ell}$ be the states at $t'$ such that $f=f_i|_{B_{t}}$ for all
 $i\in [\ell]$. 
 If there exists $i\in [\ell]$ such that 
    ${\cal D}[t',f_i]\neq \bot$, 
    then we arbitrarily choose one such  $i$ and store ${\cal D}[t,f]:= {\cal D}[t',f_i]$.
    Otherwise, we store ${\cal D}[t,f]:=\bot$. 
 
    \item \textbf{Case 5: $t$ is a join node.} Let $t$ be a join node with child nodes $t_{1}$ and $t_{2}$. We know that $B_t=B_{t_1}=B_{t_2}$. Moreover, we know that $E(G_t)=E(G_{t_1}) \cup E(G_{t_2})$ and $E(G_{t_1} \cap E(G_{t_2})=\emptyset$. Let $f$ be a state at $t$. A pair of states $(f_1,f_2)$ where $f_1$ and $f_2$ are states at $t_1$ and $t_2$, respectively, is {\em compatible} for $f$, if the following holds. For all $v\in B_t$, let $f(v)=(x_v,y_v)$, $f_1(v)=(x^{(1)}_v,y^{(1)}_v)$ and $f_1(v)=(x^{(2)}_v,y^{(2)}_v)$. Then, for all $v\in B_t$, $x_v=x^{(1)}_v=x^{(2)}_v$  and $y_{v}=y^{(1)}_v+y^{(2)}_v$. 
    If there are two pairs $(f_1,f_2)$ compatible with $f$ such that ${\cal D}[t_1,f_1]=H_1\neq \bot$ and ${\cal D}[t_1,f_1] = H_2 \neq \bot$, then we store ${\cal D}[t,f]:=H_1+H_2$    
    Here $H_2+H_1$ is the graph with vertex set $V(G_t)$ and edge set $E(H_1)\cup E(H_2)$. Otherwise we store 
    ${\cal D}[t,f]:=\bot$.
\end{itemize}

Finally we output ${\cal D}[r,f_{\emptyset}]$ as the output where $r$ is the root node. 
By induction, one can prove that above computation is correct. That is, for any node $t$ and any state $f$ at $t$, if ${\cal D}[t,f]=H_f\neq \bot$, then $H_f$ is a partial solution of state $f$ at node $t$. Otherwise, there is no partial solution of state $f$ at node $t$. For a node $t$, the number of functions from $B_t$ to $\{0,1,\ldots,\Delta\}\times \{0,1,\ldots,\Delta\}$ is $(\Delta+1)^{2|B_t|}$. This implies that number of states at a node $t$ is at most $(\Delta+1)^{2\cdot tw}$.
Therefore we get the following theorem.

\begin{theorem}
    Let $G$ be a given graph, with its nice tree decomposition of treewidth at most $tw$. Then a proper weight function for \vcew can be computed in time $O((\Delta+1)^{4(tw+1)}n)$, if it exists.
\end{theorem}

The bottleneck in the running time is the computation at a join node $t$. For a pair of states $(f_1,f_2)$, we need to check whether it is compatible with state $f$ at $t$. So the running time is upper bounded by $O((\Delta+1)^{4(tw+1)}n)$. 

The above algorithm can be modified to get an algorithm for \prvcew with same running time. The change is as follows. In the introduce edge node, we remove step 2 if the introduced edge is pre-weighted with $1$ and we remove step 3 if the introduced edge is pre-weighted with $0$.


\section{Hardness by the size of a minimum feedback vertex set}

\label{sec:hardness}

To prove \vcew{} parameterized by the size of feedback vertex set is $W[1]$-hard, we give a parameterized  reduction 
from \lc{} parameterized by the vertex cover number.

\defproblem{\lc}{: Graph $G=(V,E)$ and for each vertex $v$, a list $L(v)$ of colors.} {Find a proper coloring $c$ of $G$ such that $c(v) \in L(V)$, for all $v \in V(G)$.}

It is known that \lc{} parameterized the vertex cover number is $W[1]$-hard \cite{fiala2011parameterized}.  
Recall that an instance of  \lc{} is $(G,{\cal L}=(L(v))_{v\in V(G)})$ where $G$ is an $n$-vertex graph and $L(v)$ is a set of colors for each vertex $v\in V(G)$. We also assume that the size of each list $L(v)$ is at most $n$. Because, if $|L(v)|> n$, then in any proper list coloring of $G$, at least $|L(v)|- (n-1)$ colors are not used to color the neighbors of $v$ in $G$ and so $G$ has a proper list coloring if and only if $G-v$ has a proper listing coloring with the same set of lists provided for each vertex $u \in V(G-v)$. 
Thus, the number of colors present in the union of all lists is at most $n^2$. So, without loss of generality, we assume that for each vertex $v\in V(G)$, $L(v)\subseteq \{2,\ldots,n^2+1\}$. That is, $1$ is not a color in the lists and the maximum color is $n^2+1$.  
Let $t=\max (\bigcup_{v\in V(G)} L(v))$.


Before explaining the reduction, we describe a few gadgets we use. The first gadget is a {\em suspended path} of length $2$ which is a graph on three vertices $v,x$ and $y$, and two edges $vx, xy$. 
In this paper, we always use a suspended path of length two, so unless specified, it is always of length two. In our reduction, we use a suspended path in such a way that $x$ and $y$ will not have any other neighbors, but $v$ may have other neighbors in the reduced graph.  
For a weight function $w$, if $w(vx)=0$, then ${\sf color}_w(x)=0+ w(xy) = w(xy)={\sf color}_w(y) $, and therefore the weight function $w$ is not proper. So, whenever there is such a path in a graph, any proper weight function forces $w(vx)$ to be $1$. If there are $k$ suspended paths at $v$, then the color of $v$ is at least $k$ for any proper weight function. See Figure~\ref{tw} for an illustration. 

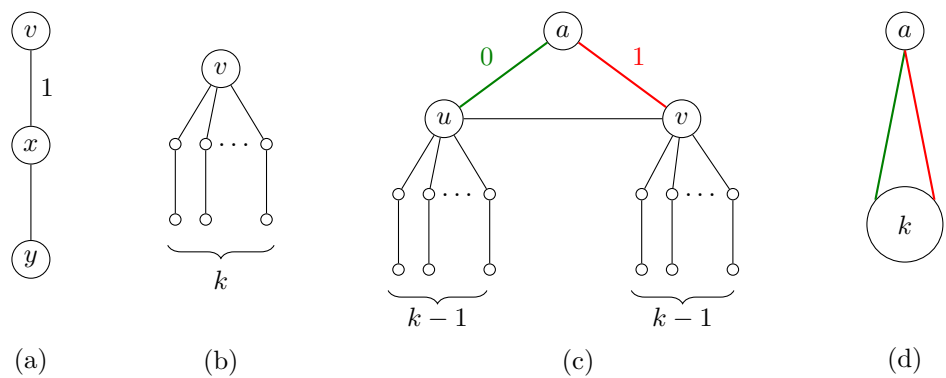
\begin{figure}[h]
    \centering

\usetikzlibrary{decorations.pathreplacing, positioning}

\begin{tikzpicture}[ 
    node distance=1.2cm and 0.8cm,
    mynode/.style={circle, draw, minimum size=0.5cm, inner sep=0pt},
    smallnode/.style={circle, draw, minimum size=0.15cm, inner sep=0pt}
]

\begin{scope}[local bounding box=a]
    \node[mynode] (u_a) {$v$};
    \node[mynode, below=1 cm of u_a] (x_a) {$x$};
    \node[mynode, below=1cm of x_a] (y_a) {$y$};
    
    \draw (u_a) -- node[right] {1} (x_a);
    \draw (x_a) -- (y_a);
    \node[below=0.8cm of y_a] {(a)};
\end{scope}

\begin{scope}[shift={(2.5,-0.5)}, local bounding box=new_b]
    \node[mynode] (v_new) {$v$};
    
    \foreach \x/\i in {-0.6/1, -0.2/2, 0.6/3} {
        \node[smallnode] (mid\i) at ([shift={(\x,-1)}]v_new.center) {};
        \node[smallnode] (leaf\i) at ([shift={(0,-1)}]mid\i.center) {};
        \draw (v_new) -- (mid\i);
        \draw (mid\i) -- (leaf\i);
    }
    \node at ([shift={(0.2,-1)}]v_new.center) {$\dots$};
    
    \draw [decorate, decoration={brace, amplitude=5pt, mirror, raise=4pt}]
        (-0.7,-2.2) -- (0.7,-2.2) node [black, midway, yshift=-0.6cm] {$k$};
        
    \node at (0,-3.9) {(b)};
\end{scope}

\begin{scope}[shift={(7,0)}, local bounding box=c]
    \node[mynode] (a_b) {$a$};
    \node[mynode, below left=0.8cm and 1.2cm of a_b] (u_b) {$u$};
    \node[mynode, below right=0.8cm and 1.2cm of a_b] (v_b) {$v$};
    
    \draw[black!50!green, thick] (a_b) -- node[above left] {0} (u_b);
    \draw[red, thick] (a_b) -- node[above right] {1} (v_b);
    \draw (u_b) -- (v_b);
    
    \foreach \x/\i in {-0.6/1, -0.2/2, 0.6/3} {
        \node[smallnode] (uc\i) at ([shift={(\x,-1)}]u_b.center) {};
        \node[smallnode] (ucc\i) at ([shift={(0,-1)}]uc\i.center) {};
        \draw (u_b) -- (uc\i);
        \draw (uc\i) -- (ucc\i);
    }
    \node at ([shift={(0.2,-1)}]u_b.center) {$\dots$};
    
    \draw [decorate, decoration={brace, amplitude=5pt, mirror, raise=4pt}]
        (-2.3,-3.3) -- (-1,-3.3) node [black, midway, yshift=-0.5cm] {$k-1$};

    \foreach \x/\i in {-0.6/1, -0.2/2, 0.6/3} {
        \node[smallnode] (vc\i) at ([shift={(\x+3.2,-1)}]u_b.center) {};
        \node[smallnode] (vcc\i) at ([shift={(0,-1)}]vc\i.center) {};
        \draw (v_b) -- (vc\i);
        \draw (vc\i) -- (vcc\i);
    }
    \node at ([shift={(3.4,-1)}]u_b.center) {$\dots$};
    
    \draw [decorate, decoration={brace, amplitude=5pt, mirror, raise=4pt}]
        (0.9,-3.3) -- (2.25,-3.3) node [black, midway, yshift=-0.5cm] {$k-1$};
        
    \node at (0.2,-4.4) {(c)};
\end{scope}

\begin{scope}[shift={(11.5,0)}, local bounding box=d]
    \node[mynode] (a_c) {$a$};
    \node[mynode, minimum size=1cm, below=1.8cm of a_c] (k_c) {$k$};
    
    \draw[black!50!green,thick] (a_c.south) -- (k_c.140);
    \draw[red,thick] (a_c.south) -- (k_c.40);
    
    \node[below=1.0cm of k_c] {(d)};
\end{scope}

\end{tikzpicture}

\caption{(a) Suspended path, (b) $k$ suspended paths at $v$, (c) Type-A $k$-disallowing gadget, (d) Its representation}
 \label{tw}
\end{figure}
%

Next, we define our second gadget, which is 
{\em $k$-disallowing gadget of type-A},  for any integer $k\geq 2$. 
Consider a $C_3$ of vertices $a, u,v$ as shown in Figure~\ref{tw}(c). Both $u$ and $v$ have $k-1$ suspended paths. In our reduction, we use this gadget in such a way that the vertex $a$ will have neighbors in addition to $u,v$, and other vertices in the gadget will not have any neighbor other than the neighbors mentioned in the gadget. If a weight function $w$ assigns $w(au)=w(av)$, then ${\sf color}_{w}(u)=w(au)+w(uv)+(k-1)=w(av)+w(uv)+(k-1)= {\sf color}_{w}(v)$, which is a coloring conflict. Hence, $w$ is not a proper weight function. So, any proper weight function $w$ forces $w(au) \neq w(av)$. In a proper weight function $w$, we may assume that $w(au)=0$ and $w(av)=1$. An edge of $C_3$ that is forced to have weight zero or one is depicted with a green or red edge, respectively, in Figure \ref{tw} (c). If $w(uv)=0$, then ${\sf color}_w(u)=k-1$, and ${\sf color}_w(v)=k$. If $w(uv)=1$, then ${\sf color}_w(u)=k$, and ${\sf color}_w(v)=k+1$. In any case, one of the neighbors of $a$ gets a color $k$. Hence, any proper weight function should not induce color $k$ on the vertex $a$. While this gadget disallows color $k$ for $a$, it also forces it to have an edge of weight $1$ incident on it. This type-A gadget is represented in short as shown in Figure~\ref{tw}(d). We use this representation in the construction of our reduction.

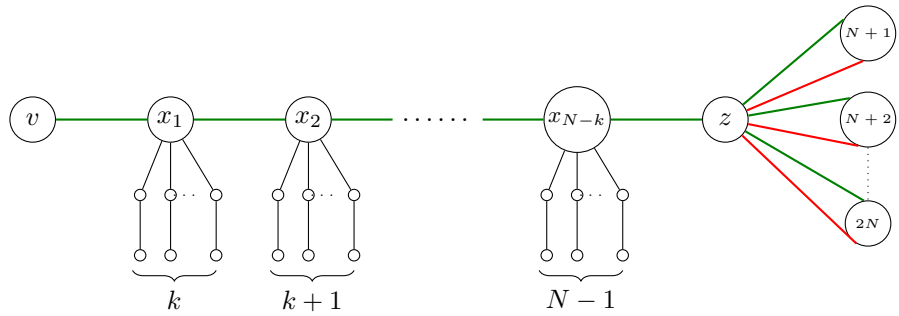
\begin{figure}[h]
    \centering
 
\begin{tikzpicture}[
    node distance=1.5cm and 1.2cm,
    mynode/.style={circle, draw, minimum size=0.6cm, inner sep=1pt},
    ellnode/.style={ellipse, draw, minimum width=1.2cm, minimum height=0.7cm, inner sep=1pt},
    smallnode/.style={circle, draw, fill=white, minimum size=0.15cm, inner sep=0pt},
    smallnode2/.style={circle, draw, fill=black, minimum size=0.15cm, inner sep=0pt}
]

\node[mynode] (v) {$v$};

\node[mynode, right=of v] (x1) {$x_1$};
\node[mynode, right=of x1] (x2) {$x_2$};

\draw[black!50!green, thick] (v) -- (x1) -- (x2);

\foreach \n/\label in {x1/k, x2/k+1} {
    \foreach \x/\i in {-0.4/1, 0/2, 0.6/3} {
        \node[smallnode] (c\n\i) at ([shift={(\x,-1)}]\n.center) {};
        \node[smallnode] (cc\n\i) at ([shift={(0,-0.8)}]c\n\i.center) {};
        \draw (\n) -- (c\n\i);
        \draw (c\n\i) -- (cc\n\i);
    }
    \node at ([shift={(0.2,-1)}]\n.center) {\tiny $\dots$};
    \draw [decorate, decoration={brace, amplitude=4pt, mirror}]
        ([shift={(-0.5,-2)}]\n.center) -- ([shift={(0.6,-2)}]\n.center) 
        node [midway, below=4pt] {$\label$};
}

\node[right=0.8cm of x2] (dots) {$\dots \dots$};
\node[mynode, right=0.8cm of dots] (xnk) {\footnotesize $ x_{N-k}$};
\draw[black!50!green, thick] (x2) -- (dots) -- (xnk);

\foreach \x/\i in {-0.4/1, 0/2, 0.6/3} {
    \node[smallnode] (cnk\i) at ([shift={(\x,-1)}]xnk.center) {};
    \node[smallnode] (ccnk\i) at ([shift={(0,-0.8)}]cnk\i.center) {};
    \draw (xnk) -- (cnk\i);
    \draw (cnk\i) -- (ccnk\i);
}
\node at ([shift={(0.2,-1)}]xnk.center) {\tiny $\dots$};
\draw [decorate, decoration={brace, amplitude=4pt, mirror}]
    ([shift={(-0.5,-2)}]xnk.center) -- ([shift={(0.6,-2)}]xnk.center) 
    node [midway, below=4pt] {$N-1$};

\node[mynode, right=of xnk, minimum size=0.6cm] (z) {$z$};
\draw[black!50!green, thick] (xnk) -- (z);

\node[mynode, right=1.2cm of z] (n2) {\tiny $N+2$};
\node[mynode, above=0.4cm of n2] (n1) {\tiny $N+1$};
\node[mynode, below=0.7cm of n2] (nn) {\tiny $2N$};

\draw[black!50!green, thick] (z) -- (n1.150);
\draw[red, thick] (z) -- (n1.260);
\draw[black!50!green, thick] (z) -- (n2.130);
\draw[red, thick] (z) -- (n2.250);
\draw[black!50!green, thick] (z) -- (nn.100);
\draw[red, thick] (z) -- (nn.240);

\draw[dotted] (n2) -- (nn);

\end{tikzpicture}
    \caption{$k$-disallowing gadget of type-B at $v$.}
    \label{fvs1}
\end{figure}


Notice that there is a triangle in the $k$-disallowing gadget of type-A. Since we want the feedback vertex set size to be bounded by a function of the parameter, in our reduced graph, we can not use {\em many} disjoint copies of $k$-disallowing gadgets of type-$A$.   
 So, with the help of the type-A gadget, we design a new gadget, called the type-B $k$-disallowing gadget, $k \geq 2$.

Here, we use a universal vertex $z$ and add a large number of type-A disallowing gadgets only at $z$. The vertex $z$, and its type-A gadgets are shared by all the type-B gadgets in the construction. 
Figure \ref{fvs1} shows such a gadget at vertex $v$. Here, $N$ is a large number greater than $k$ which we will explain later.  
We'll shortly prove that this gadget disallows color $k$ by any proper weight function of $G$.
If vertex $v$ demands another disallowing color, say $p$, then we add a path starting from $v, y_1, y_2, \dots , y_{N-p}, z$. There will be $p$ suspended paths at $y_1$, $p+1$ suspended paths at $y_2$, and so on. Vertex $y_{N-p}$ will have $N-1$ suspended paths. And it will be adjacent to the universal vertex $z$. Figure \ref{fvs2} shows $k,p$ disallowing gadgets at $v$. Both these gadgets share the vertex $z$, and type-A gadgets incident on $z$. The shared vertices are shown by a blue rectangle in Figure~\ref{fvs2}.
Similarly, each disallowed color will create a path to reach the common vertex $z$. The same can be replicated for other vertices as well. Each gadget starts with different chain vertices and merges at $z$. 

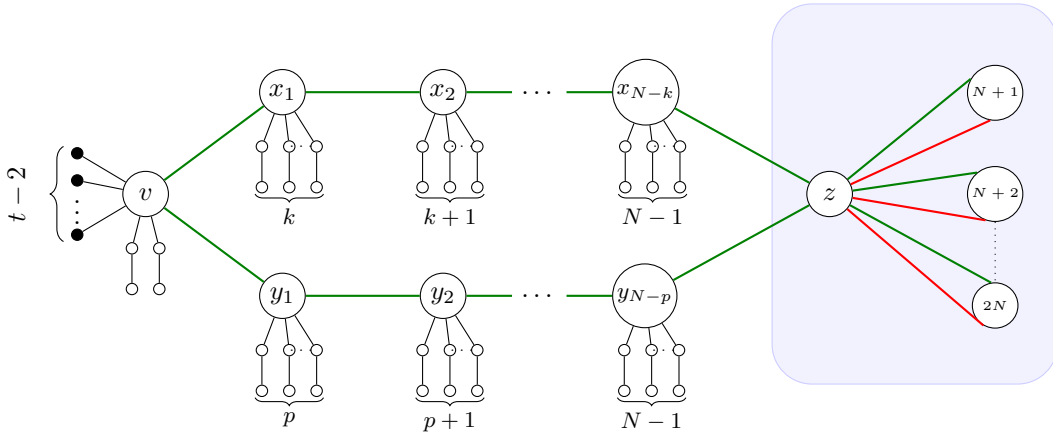
\begin{figure}[h]
    \centering
    
\begin{tikzpicture}[scale=0.9,
    node distance=1.2cm and 1.5cm,
    mynode/.style={circle, draw, minimum size=0.6cm, inner sep=1pt, fill=white},
    smallnode/.style={circle, draw, fill=white, minimum size=0.15cm, inner sep=0pt},
    smallnode2/.style={circle, draw, fill=black, minimum size=0.15cm, inner sep=0pt}
]

\node[mynode, minimum size=0.6cm] (z) at (10, 0) {$z$};

\node[mynode] (v) at (0,0) {$v$};

\foreach \y/\i in {0.6/1, 0.2/2, -0.6/3} {
    \node[smallnode2] (vtl\i) at ([shift={(-1,\y)}]v) {};
    \draw (v) -- (vtl\i);
}
\node at ([shift={(-1,-0.15)}]v) {$\vdots$};
\draw [decorate, decoration={brace, amplitude=5pt}]
    ([xshift=-1.2cm, yshift=-0.7cm]v.center) -- ([xshift=-1.2cm, yshift=0.7cm]v.center) 
    node [midway, xshift=-0.6cm, rotate=90] { $t-2$};

\foreach \x/\i in {-0.2/1, 0.2/2} {
    \node[smallnode] (vd\i a) at ([shift={(\x,-0.8)}]v.center) {};
    \node[smallnode] (vd\i b) at ([shift={(0,-0.6)}]vd\i a.center) {};
    \draw (v) -- (vd\i a) -- (vd\i b);
}

\begin{scope}[yshift=1.5cm]
    \node[mynode] (x1) at (2,0) {$x_1$};
    \node[mynode, right=of x1] (x2) {$x_2$};
    \node[right=0.6cm of x2] (xdots) {$\dots$};
    \node[mynode, right=0.6cm of xdots] (xnk) {\footnotesize $x_{N-k}$};
    
    \draw[black!50!green, thick] (v) -- (x1);
    \draw[black!50!green, thick] (x1) -- (x2) -- (xdots) -- (xnk);
    \draw[black!50!green, thick] (xnk) -- (z);

    \foreach \n/\label in {x1/k, x2/k+1, xnk/N-1} {
        \foreach \x/\i in {-0.3/1, 0.1/2, 0.5/3} {
            \node[smallnode] (c\n\i) at ([shift={(\x,-0.8)}]\n.center) {};
            \node[smallnode] (cc\n\i) at ([shift={(0,-0.6)}]c\n\i.center) {};
            \draw (\n) -- (c\n\i) -- (cc\n\i);
        }
        \node at ([shift={(0.3,-0.8)}]\n.center) {\tiny $\dots$};
        \draw [decorate, decoration={brace, amplitude=3pt, mirror}]
            ([shift={(-0.4,-1.5)}]\n.center) -- ([shift={(0.6,-1.5)}]\n.center) 
            node [midway, below=2pt] {\footnotesize $\label$};
    }
\end{scope}

\begin{scope}[yshift=-1.5cm]
    \node[mynode] (y1) at (2,0) {$y_1$};
    \node[mynode, right=of y1] (y2) {$y_2$};
    \node[right=0.6cm of y2] (ydots) {$\dots$};
    \node[mynode, right=0.6cm of ydots] (ynk) {\footnotesize $y_{N-p}$};
    
    \draw[black!50!green, thick] (v) -- (y1);
    \draw[black!50!green, thick] (y1) -- (y2) -- (ydots) -- (ynk);
    \draw[black!50!green, thick] (ynk) -- (z);

    \foreach \n/\label in {y1/p, y2/p+1, ynk/N-1} {
        \foreach \x/\i in {-0.3/1, 0.1/2, 0.5/3} {
            \node[smallnode] (cy\n\i) at ([shift={(\x,-0.8)}]\n.center) {};
            \node[smallnode] (ccy\n\i) at ([shift={(0,-0.6)}]cy\n\i.center) {};
            \draw (\n) -- (cy\n\i) -- (ccy\n\i);
        }
        \node at ([shift={(0.3,-0.8)}]\n.center) {\tiny $\dots$};
        \draw [decorate, decoration={brace, amplitude=3pt, mirror}]
            ([shift={(-0.4,-1.5)}]\n.center) -- ([shift={(0.6,-1.5)}]\n.center) 
            node [midway, below=2pt] {\footnotesize $\label$};
    }
\end{scope}

\node[mynode, right=1.5cm of z] (n2) {\tiny $N+2$};
\node[mynode, above=0.6cm of n2] (n1) {\tiny $N+1$};
\node[mynode, below=0.8cm of n2] (nn) {\tiny $2N$};

\draw[black!50!green, thick] (z) -- (n1.150); \draw[red, thick] (z) -- (n1.260);
\draw[black!50!green, thick] (z) -- (n2.130); \draw[red, thick] (z) -- (n2.250);
\draw[black!50!green, thick] (z) -- (nn.100); \draw[red, thick] (z) -- (nn.240);
\draw[dotted] (n2) -- (nn);

\begin{scope}[on background layer]
    \draw[fill=blue!5, draw=blue!20, rounded corners=15pt] 
        ($(z.west) + (-0.5, 2.8)$) rectangle ($(n2.east) + (0.5, -2.8)$);
\end{scope}

\end{tikzpicture}

    \caption{Illustration of two disallowing gadgets of type-B at $v$.}
    \label{fvs2}
\end{figure}

We use \lc for the parameterized reduction. Every vertex of $G$ has a list of admissible colors given. The idea is to disallow all the colors except $L(v)$ for each $v \in V(G)$. 
Suppose we want to disallow color $k$ on a vertex $ v$. Then $v$ should have a neighbor, say $x$, such that $x$ has color $k$ under any proper function.
We can force $x$ to have a particular color using suspended paths. But the edge $vx$ cannot be forced to have a particular weight unless it is part of a cycle, just as in the type-A disallowing gadget. And based on this edge weight, the color of $x$ can not be guaranteed. Our main idea for the type-B gadget is to force the edge $vx$ to have weight $0$. We achieve this using a cascading effect that starts at the universal vertex $z$. 
We force $z$ to have a very large color $N$ using type-A gadgets. So that any edge incident to $z$, apart from type-$A$ gadgets edges, must have weight $0$, or else it will cause a coloring conflict. 
Let us prove this formally. 

\begin{lemma}
\label{lem:red:cascade}
Let $D$ be the graph shown in Figure~\ref{fvs2} and $D$ is an induced subgraph of a graph $H$ such that for any edge $e$ between $V(D)$ and $V(H)\setminus V(D)$, the endpoint of $e$ in $V(D)$ is from $\{v,z\}$.  
Let $N,k,p,t\in {\mathbb N}$ such that  $t\geq 2, N>k,p$ and $d_H(z)\leq 3N$. 
Let $w$ be a proper weight function of $H$. 
Then, $w(vx_1)=w(vy_1)=0$, ${\sf color}_w(x_1)=k$ and ${\sf color}_w(y_1)=p$. 
\end{lemma}

\begin{proof}
First, we prove that any edge incident on $z$ which is not from any type-$A$ disallowing gadget must have weight $0$. This will eventually help us show $w(vx_1)=0$, and hence ${\sf color}_w(x_1)=k$. Similar arguments show that $w(vy_1)=0$ and ${\sf color}_w(y_1)=p$, so we don't explicitly mention their proofs.

There are $N$ type-$A$ gadgets incident on $z$, as shown in Figure~\ref{fvs2}, with exactly two edges incident on it from each type-$A$ gadget, one with weight $0$, and other with weight $1$. So, color of $z$ induced by any proper weight function is at least $N$. 
Given that $d_H(z) \leq 3N$, there are at most $N$ edges incident on it, apart from type-$A$ gadget edges. 
We claim that any such edge must have weight $0$. If not, then there are some edges incident on $z$ with weight $1$. Hence, the color induced on $z$ will be in the range $[N+1,2N]$, as there are at most $N$ such edges incident on $z$, other than the edges in the type-A gadgets.
Vertex $z$ is disallowed colors $[N+1,2N]$ by the type-$A$ gadgets.
So, the only feasible color for $z$ under any proper weight function is $N$.
This means every edge incident on $z$, not part of type-$A$ gadget, is forced to have weight $0$ under any proper weight function. Hence, $w(zx_{N-k})=0$ and $w(zy_{N-p})=0$. 

Consider vertex $x_{N-k}$. It has $N-1$ suspended paths incident on it, causing its color at least $N-1$. We have already proved that $w(zx_{N-k})=0$. Hence color induced on $x_{N-k}$ under any proper weight function is $N-1+0+w(x_{N-k}x_{N-k-1})$. If  $w(x_{N-k}x_{N-k-1})=1$, then ${\sf color}_w(x_{N-k})=N$, a conflict with ${\sf color}_w(z)$. 
To avoid such a conflict, any proper weight function must assign $w(x_{N-k}x_{N-k-1})=0$, and hence ${\sf color}_w(x_{N-k})=N-1$.
Similarly, $w(x_{N-k-1}x_{N-k-2})$ must be $0$. If not, then ${\sf color}_w(x_{N-k-1})=(N-2)+1$ (here $N-2$ comes from suspended paths at $x_{N-k-1}$), a conflict between $x_{N-k}$ and $x_{N-k-1}$.
In the same fashion, we claim that $w(x_ix_{i-1})=0$, for $i$ in $\{N-k-2, N-k-3, \dots ,2\}$. So, $w(x_2x_1)=0$, and ${\sf color}(x_2) =k+1$. There are $k$ suspended path incident on $x_1$. If $w(vx_1)=1$, then ${\sf color}_w(x_1) = k+1 = {\sf color}_w(x_2)$, a conflict. Hence, $w(vx_1)=0$ and ${\sf color}_w(x_1)=k$.
\end{proof}

\noindent
{\bf Construction.}
We have explained all the gadgets needed for the reduction. Now we will explain the reduction. 
Let $(G, \mathcal{L})$ be an instance of \lc, and $n=|V(G)|$. Recall that 
for each vertex $v\in V(G)$, $L(v)\subseteq \{2,\ldots,n^2+1\}$. That is, $1$ is not a color in the lists and the maximum color is $n^2+1$.  
Let $t=\max (\bigcup_{v\in V(G)} L(v))$. Let $N=n^3+n^2-n$. 
We construct a graph $H$ as follows. Initially, we set $H:=G$. Now we add gadgets to it. 

\begin{enumerate}
    \item For each vertex $v \in V(G)$, add $k$-disallowing gadgets of type-B for all $k\in [t+n-1] \setminus (L(v) \cup \{1\})$. As mentioned earlier, the type-A gadget (i.e., the vertices in blue rectangle in Figure~\ref{fvs2}) is common for all the type-B gadgets. 
    \item Add $t-2$ pendant vertices to each vertex of $G$. Note that $t \geq 2$.
    \item Add two suspended paths of length $2$ at each vertex $v \in V(G)$.
\end{enumerate}

This completes the construction of the reduced graph $H$. 
Figure \ref{fvs2} shows the construction at $v$ with disallowed colors $k$ and $p$.
%
Before we prove the correctness of the reduction, let us discuss the selection of the value $N=n^3+n^2-n$. Recall that $z$ is a vertex in $H$ which is common in all the $k$-disallowing gadgets of type-B used in the reduction. There are $n$ vertices in $G$ and we want to disallow at most $t+n-2$ colors at vertex in $G$ (because $|L(v)| \geq 1, \forall v \in V(G))$. So, the number of edges incident to $z$ that are not part of the {\em common} type-A gardget part, is at most $n(t+n-2)\leq n^3+n^2-n$.  So, we choose $N=n^3+n^2-n$ and this implies that $d_H(z)\leq 3N$, a condition needed for Lemma~\ref{lem:red:cascade}. 

\begin{lemma}
\label{lem:correctnessfvs}
$(G, \mathcal{L})$ is yes-instance of \lc if and only if $H$ is yes-instance of \vcew.
\end{lemma}
    \begin{proof}
  Let $(G, \mathcal{L})$ be a yes-instance of \lc, and the corresponding proper coloring be $c~:~V(G)\rightarrow [t]$. We define a weight function $w$ on $E(H)$ based on $c$ as follows. Assign $w(e) = 0$, for all $e \in E(G)$. 
  Now we need to assign weights to edges that are exclusively in $H$.
  In the construction, we added $t-2$ pendant vertices at every vertex of $G$. 
  For a vertex $v \in V(G)$, if $c(v) = p \in L(v)$, assign weight one to $p-2$ of the $t-2$ pendant edges and zero to the remaining $t-p$ pendant edges. Note that $2 \leq p \leq t$.
  There are two suspended paths at each vertex of $G$. Assign weight zero to the pendant edges and weight one to the remaining edges. In fact, we do the same for all the suspended paths in $H$. Assign weight zero to the pendant edge of it, and weight one to the remaining edge of it.
  Every type-B gadget contains a path that joins the vertex of $G$ to $z$. 
  As observed earlier, all path edges in the type-B gadgets must have weight $0$ (green edges in Figure~\ref{fvs2}). So we assign weight $0$ for all the path edges.
  Every vertex in the path has suspended paths incident on it, and we mentioned their weight values. 
  The type-A disallowing gadgets have triangles, and there are two edges from each triangle incident on $z$. They are already forced to have weights $1$ and $0$, respectively. And we assign them those weights. These edges are shown in color red and green in Figure~\ref{fvs2}, respectively. 
  We assign a weight of one to the remaining third edge in each triangle. The suspended paths in this gadget also follow the weight values we mentioned before. With this, our description of the weight function $w$ is complete; we have defined it for all the edges in $H$.
  \begin{claim}
      The weight function $w$ is proper, and ${\sf color}_w(v)=c(v)$, for all $v \in V(G)$. 
  \end{claim}

  \begin{claimproof}
      For a vertex $v \in V(G)$, all the edges incident on it except the $c(v)-2$ pendant edges are assigned weight zero. But there are two suspended paths at $v$ contributing $2$ to its color induced by $w$. So, ${\sf color}_w(v)=c(v)-2+2 =c(v)$. Given that $c$ is a proper coloring for $G$, any two adjacent vertices in $G$ get distinct colors. As color induced by $w$ is same as $c(v)$, $w$ is proper on the vertices of $G$.
      Now it remains to show that $w$ is proper on rest of the vertices as well. The $t-2$ pendant vertices at $v$ get color either $1$ or $0$. And these vertices do not cause coloring conflict as $0,1 \notin L(v)$, for all $v \in V(G)$.
      There are two suspended paths at $v$, and color of the middle vertex is $1$, while the pendant vertex gets color $0$. So, such vertices do not create any conflict.

      We consider one type-B disallowing gadget and prove that $w$ doesn't cause a coloring conflict with any vertex in this gadget. The similar arguments hold true for each such gadget, and the proof follows. 
      Consider the gadget where $v$ is disallowed a color $k$. For the  first vertex in the path (i.e., the vertex adjacent to $v$), ${\sf color}_w(x_1)=k$, and $k \notin L(v)$, while $c(v)={\sf color}_w(v) \in L(v)$. So $w$ doesn't cause coloring conflict of $v$ and $x_1$. 
      The vertex $x_1$ has $k$ suspended paths incident on it. Let $a$ and $b$ be vertices of one such suspended path. That is $ax_1, ab \in E(H)$. As described above, ${\sf color}_w(a) =1$, and ${\sf color}_w(b)=0$. Note that $k \geq 2$. So, $w$ is proper on this suspended path. And so it does hold for every suspended path incident on $x_1, x_2, \dots , x_{N-k}$.  
      The path vertices $x_1, x_2, \dots , x_{N-k}$ get colors $k, k+1, \dots N-1$, respectively, without  conflicting with each other.


      Now consider the type-A gadgets at $z$. We have $k$-disallowing gadgets of type-A at $z$ for all $k\in \{N+1,\ldots,2N\}$. Consider the 
      ($N+i$)-disallowing gadget at $z$ for $i\in [N]$. 
      Let the triangle in this gadget be  $z,s_i,t_i$ where $w(zs_i)=0$,  $w(zt_i)=1$ and $w(s_it_i)=1$. So, ${\sf color}_w(s_i)=N+i$, and ${\sf color}_w(t_i)=N+i+1$. 
      For the suspended paths attached to vertices $s_i$ and $t_i$, the pendant vertices gets color $0$, and the middle vertices gets color $1$, and $N+i, N+i+1 >1$, so it doesn't cause coloring conflict.  Notice that there are $N$ type-A gadgets at $z$ and exactly $N$ edges incident of $z$ are set to one by $w$. So, we have 
      ${\sf color}_w(z)=N$ and therefore, no coloring conflict between $z$ and any of its neighbors.       So we have proved that for any two adjacent vertices in $H$, the colors induced by $w$ are distinct. Hence $w$ is a proper weight function.
  \end{claimproof}

  Hence, if  $(G,\mathcal{L})$ is a yes-instance of \lc, then there is a proper weight function $w$ on $H$,  making $H$ a yes-instance of \vcew.

  Conversely, we prove that if $H$ is a yes-instance of \vcew, then $(G,\mathcal{L})$ is a yes-instance of \lc. 
  Assume that there is a proper weight function $w$ on $E(H)$. 
  Consider a vertex $v \in V(G) \cap V(H)$. Its degree $d_H(v) = d_G(v)+t+(t+n-1)-|L(v)|-1$. Because $v$ is incident with $t-2$ pendant vertices, and two suspended paths. It has $(t+n-1)-|L(v)|-1$ type-B disallowing gadgets incident on it. 
 By Lemma~\ref{lem:red:cascade}, all the $t+n-1-L(v)-1$ type-B gadget edges incident on $v$  are forced to have weight zero.
  Hence, the maximum value of color for $v$ that can be obtained under any proper weight function is $d_G(v)+t$. Also, because of the two suspended paths, color of $v$ under any proper weight function is at least $2$. Hence, $$2 \leq {\sf color}_w(v) \leq d_G(v)+t \leq n-1+ t .$$ 
  But, we added disallowing gadgets at $v$ for all the colors $[t+n-1]\setminus(L(v) \cup \{1\})$. Also, because of its two suspended paths, it can not have colors $0$ or $1$. 
  Thus, the only feasible colors of $v$ are from $L(v)$. The weight function $w$ is proper, so no two adjacent vertices have the same color.  
  Define a coloring $c$ as $c(v) = {\sf color}_w(v)$,  for all $v \in V(G)$. 
  If $uv \in E(G)$, then $uv \in E(H)$. As $w$ is proper, ${\sf color}_w(u) \neq {\sf color}_w(v)$. So, $c(u) \neq c(v)$ for any two adjacent vertices in $G$.
  And we have already explained that $c(v) \in L(v)$ for all $v \in V(G)$. Hence, $c$ is in fact a proper list coloring of $G$.
\end{proof}
\begin{theorem}\label{thm:fvs}
    \vcew is $W[1]$-hard parameterized by the size of a minimum feedback vertex set.
\end{theorem}
\begin{proof}
  \lc is $W[1]$-hard parameterized by vertex cover number. 
  The construction of $H$ takes polynomial time in $n$. 
  By lemma \ref{lem:correctnessfvs}, $(G, \mathcal{L})$ is yes-instance of \lc if and only if $H$ is yes-instance of \vcew.
    Let $vc$ be size of a minimum vertex cover $S$ of $G$. Now we claim that $S\cup \{z\}$ is a feedback vertex set in $H$. Observe that any cycle in $H-z$ is also a cycle in $G$. Since $S$ is a vertex cover in $G$, $S$ hits all the cycles in $G$. So $S\cup \{z\}$ hits all the cycles in $H$ and thus $H$ has a feedback vertex set of size at most $vc+1$. 
    Hence, \vcew is $W[1]$-hard parameterized by the size of a minimum feedback vertex set.  
\end{proof}

Since treewidth of a graph is at most the size of a minimum feedback vertex set in it, \vcew is also $W[1]$-hard, when parameterized by the treewidth. 
Theorem~\ref{thm:fvs} also implies that \prvcew{} is $W[1]$-hard by the size of a minimum feedback vertex set.


\section{\prvcew}
\label{sec:prewt}
A general version of \vcew is defined as follows. Given a graph $G$, a subset $E_1$ of edges of $G$, and a pre-edge weighting function ${w}: E_1 \to \{0,1\}$, find a proper weight function $\hat{w}$ that extends the pre-edge weighting function, if it exists. 
\begin{corollary}
    \prvcew is \WoneHard parameterized by the size of a minimum feedback vertex set.
\end{corollary}
\prvcew boils down to \vcew when $E_1$ is empty set, and hence \vcew is a special case of \prvcew. By theorem \ref{thm:fvs}, \vcew is \WoneHard parameterized by $fvs$. Hence, the corollary.



We consider a restricted version of \prvcew in which the pre-edge weights are of a single type. That is the pre-edge weighting function ${w}:E_1 \to \{1\}$. 
We prove that this restricted version is tractable when parameterized by the vertex cover number.
Here onwards, whenever we mention \prvcew, it is this restricted version.
\subsection{FPT algorithm parameterized by Vertex Cover number}
Let $G$ be a given graph, with vertex cover number at most $k$. In case of a yes-instance of \vcew, we gave an upper bound $8k^2+8k=T$ on the colors induced by a proper weight function.
The same bound doesn't hold true anymore for \prvcew. The reason is that there can be many edges ($> T$) pre-weighted by a given function ${w}$, which can not be changed. However, we can still give a bound on the additional gain in the colors, apart from the color induced by ${w}$. 
Given ${w}:E_1 \to \{1\}$, we define $base_{{w}}(v)$ as the base color of vertex $v$ induced by ${w}$. 
If $(G,{w},k)$ is a yes-instance, then there exists a proper weight function $\hat{w}$ such that ${w}(e)=\hat{w}(e), \forall e \in E_1$. In this case, we call $\hat{w}$ as an extension of $w$. And ${\sf color}_{\hat{w}}(v) \geq  base_{{w}}(v)$ for every $v \in V(G)$. We give a bound on ${\sf color}_{\hat{w}}(v) - base_{{w}}(v)$.

\begin{lemma}
\label{lem:precolor}
   Let $G$ be a graph with vertex cover number $k$ and $(G,{w},k)$ is a yes-instance \prvcew, with all the pre-weights are 1. 
    Then there exists a proper weight function $\hat{w}:~E(G)\rightarrow \{0,1\}$, which is an extension of ${w}$ and for all $v\in V(G)$,
    ${\sf color}_{\hat{w}}(v) - base_{{w}}(v) \leq 8k^2+8k$.  
\end{lemma}
\begin{proof}[Proof sketch]
The proof of Lemma \ref{lem:precolor} follows similar arguments as in the proof of Lemma \ref{lem:boundedcolor}. So, instead of writing the complete proof, we give a brief description of the proof structure.  
For any proper weight function $w'$ of $G$, which is an extension of ${w}$, we define a { \textit{potential}} $P(w')$ as   
$P(w')=\sum_{v \in V(G)}\max (0, {\sf color}_{w'}(v)-base_{{w}}(v)-T)$. 
Let $S$ be a minimum vertex cover of $G$, and $I = V(G) \setminus S$. Given that $|S| =k$, color of every vertex in $I$ is at most $k$.
 
Let $\hat{w}$ be a weight function with the minimum  \textit{potential} possible.
If $P(\hat{w}) =0$, then the additional gain in the color after ${w}$ of every vertex is bounded by $T$, and $\hat{w}$ is the desired weight function.
If not, for contradiction, we assume that $P(\hat{w}) >0$, and note that $\hat{w}$ is a proper weight function and is an extension of ${w}$.
Sets $C$ and $D$ are the same as mentioned in the proof of Lemma \ref{lem:boundedcolor}.
Consider a vertex $x \in S$, ${\sf color}_{\hat{w}}(x)={c}$ such that ${c}-base_w(v) > T$. 
This means that apart from the edges assigned weight one by ${w}$, there are more than $T$ edges assigned weight one exclusively by $\hat{w}$. Let $Y$ be the set of endpoints of such edges. 
We emphasize again that a vertex $v$ such that ${w}(vx)=1$ is not part of the set $Y$, and $|Y| \geq 8k^2+7k+2$.
Here onwards, the proof is exactly the same as the proof of Lemma \ref{lem:boundedcolor}. 

We describe the weight function $w'$ again for reference. Note that we started with the original proper weight function $\hat{w}$, which is an extension of $w$.
\begin{itemize}
    \item For every $y \in Y_1$, set $w'(xy) = 0 \text{ (was }1)$ and $w'(zy) = 1 \text{ (was }0)$.
    \item Leave all other edge-weights unchanged; for all other edges $e$, set $w'(e) = \hat{w}(e)$.
\end{itemize}
We need to show that weights of edges in $E_1$ are unchanged. Note that any vertex $v$ such that $e=xv \in E_1$, $v \notin Y$. And $Y_1 \subseteq Y$. So $xy \notin E_1$, whose weight is changed. 
For edge $zy$, with $\hat{w}(zy)=0$, its weight is changed to $1$. But $\hat{w}$ is an extension of $w$, and if $\hat{w}(zy)=0$, then it clearly means that $zy \notin E_1$. And the remaining edge weights remain unchanged. Hence, all edges $e \in E_1$, their weights remain $1$ in $w'$. Hence $w'$ is an extension of $w$, it is a proper weight function by the proof of Lemma \ref{lem:boundedcolor}.  
The new proper weight function for $G$ is $w'$, and $P(w') < P(\hat{w})$ which is a contradiction to the minimality of the potential of $\hat{w}$.   
\end{proof}

Next, we describe an FPT algorithm to solve \prvcew. Let $(G,w,k)$ be a given instance. Here $w:E_1 \to \{1\}$ is a pre-weighting function given, where $E_1 \subseteq E(G)$. Let $S$ be a minimum vertex cover of $G$, and $|S| =k$. Such a vertex cover can be computed in $O^*(1.25284^k)$ time, where $O^*$ hides the polynomial factors of the input size \cite{harris2024faster}. Any vertex in $I = V(G) \setminus S$ has its degree bounded by $k$. So, $base_w(u) \in \{0,1, \dots , k\}$ for all $u \in I$. 
We want to partition $I$ into equivalence classes based on their neighbors in $S$, but two vertices in such an equivalence class may not be identical, as their base colors could be different. Hence, we refine the relation on the vertices of $I$ as follows. For two vertice $u, v \in I$. $u$ is related to $v$ if two conditions hold true: (i) $N(u) =N(v)$, (ii) for every $x \in N(u)$, if $ux \in E_1$, then $vx \in E_1$ (by the first condition $N(u)=N(v)$). This new relation is also an equivalence relation, partitioning $I$ into equivalence classes. 
For a subset 
 $S_2 \subseteq S_1 \subseteq S$, with $|S_1|=p \leq k$, let $A_{S_1,S_2}$ denote an equivalence class such that for all $u \in A_{S_1,S_2}$, $N(u)=S_1$, and for every $x \in S_2$, $ux \in E_1$, and for every $y \in S_1 \setminus S_2$, $uy \notin E_1$. Hence $base_w(u)=|S_2|$.
The total number of equivalence classes is at most $3^k$.
\begin{claim}\label{claimpre}
    If there is a large equivalence class $A_{S_1,S_2}$ with $|A_{S_1,S_2}| > kT+1$, then all but $kT+1$ vertices will get color $|S_2|$ by any proper weight function that extends $w$ (if it exists), and satisfies Lemma \ref{lem:precolor}.
\end{claim}
\begin{claimproof}
Let $|S_2|=p$. For contradiction, assume there is a proper weight function $w'$ that extends $w$ and satisfies Lemma \ref{lem:precolor}, and the claim is not true. Then there will be more than $kT+1$ vertices 
$u \in A_{S_1,S_2}$ with color at least $p+1$, implying that there is at least one edge $e \notin E_1$ incident on $u$, such that $w'(e)=1$. But all the neighbors of $u$ are in $S_1$, with $|S_1| \leq k$. So, by pigeonhole principle, there will be at least one vertex $x \in S$ with ${\sf color}_{w'}(x) -base_w(x) > T$, a contradiction to the fact that $w'$ satisfies Lemma \ref{lem:precolor}. Hence, the claim holds true. 
\end{claimproof}
This claim leads to the reduction rule explained below. Here onwards, whenever we mention $w'$ is a proper weight function, it is also an extension of $w$. Hence, we don't explicitly mention that $w'$ is an extension of $w$.

Repeatedly apply the reduction rule unless for every equivalence class,  its size is at most  $kT+1$  or only pre-weighted edges are incident on the vertices in the equivalence class.

\smallskip
\noindent
\textbf{Reduction Rule}: Consider a large equivalence class $A$, with $|A| > kT+1$ and at least one edge incident on $A$ is from $E(G)\setminus E_1$. 
Let $E'(u)$ denote the set of unweighted edges (i.e., edges from $E(G)\setminus E_1$) incident on $u$. Pick $u \in A$ arbitrarily, and delete $E'(u)$.


Next we prove that the above reduction rule is safe. 
Let $H$ be the graph obtained after applying the above reduction rule once.  
That is, $H=G-E'(u)$. 
Next, we prove that $(G,w,k)$ is a yes-instance if and only if $(H,w,k)$ is a yes-instance of \prvcew. Since for any two vertices $u_1$ and $u_2$ in the same equivalence class, $(G-E'(u_1),w,k)$ is identical to $(G-E'(u_2),w,k)$ after we rename $u_2$ with $u_1$ in the later instance, by Claim~\ref{claimpre}, if $G$ has a proper weight function, then $H$ has a proper weight function. 

Conversely, suppose $H$ has a proper weight function $\hat{w}$ that satisfies Lemma~\ref{lem:precolor}. 
Here $H=G - E'(u)$, for some $u \in A$. 
We define a weight function $w'$ on $E(G)$ as follows. For all $e\in E(H)$, $w'(e)=\hat{w}(e)$,  and for all $e\in  E'(u)= E(G)\setminus E(H)$, $w'(e)=0$. 
We prove that $w'$ is a proper weight function of $G$. 
Note that ${\sf color}_{w'}(u)=base_w(u)$. 
And any vertex $x \in V(H) \setminus \{u\}$, ${\sf color}_{w'}(x)={\sf color}_{\hat{w}}(x)$,   
because all the edges in $E'(u)$ are set to zero, and the rest of the edges' weights remain unchanged.
Let $e=xy$ be an edge in $E(G)$. We have two cases. 
The first case is $e\in E(H)$, so $x, y \neq u$. Then, we know that  ${\sf color}_{w'}(x)={\sf color}_{\hat{w}}(x)$ and ${\sf color}_{w'}(y)={\sf color}_{\hat{w}}(y)$. Since $\hat{w}$ is a proper weight function, 
${\sf color}_{\hat{w}}(x)\neq {\sf color}_{\hat{w}}(y)$. Thus, we get that ${\sf color}_{w'}(x)\neq {\sf color}_{w'}(y)$.

In the second case $e=uv\in E'(u)$. Then $v \in S$. 
We know that ${\sf color}_{w'}(u)=base_w(u)$ and ${\sf color}_{w'}(v)= {\sf color}_{\hat{w}}(v)$. To prove 
${\sf color}_{w'}(v)\neq {\sf color}_{w'}(u)$, it is enough to prove that ${\sf color}_{w'}(v)\neq base_w(u)$. As we know, $u \in A$, which is a large equivalence class, with at least $kT+1$ vertices, apart from $u$. Since $\hat{w}$ satisfies Lemma~\ref{lem:precolor}, 
there is a vertex  $z\in A \cap V(H)$, such that that ${\sf color}_{w'}(z)=base_w(u)$. 
This implies that ${\sf color}_{w'}(v)\neq base_w(u)$ because $zv\in E(H)$.



Now, let $H^{\star}$ be the graph obtained after exhaustive application of the above reduction rule. Next, we upper bound the number of unweighted edges in $H^{\star}$ by $w$. There are at most $3^{k}$ equivalence classes.
Each equivalence class has at most $k(kT+1)$ unweighted edges. There are at most $k(k-1)$ unweighted edges in $S$, and $T=8k^2+8k$. So, total number of unweighted edges in $H^{\star}$ is at most $k(k-1)+3^k (k(8k^2+8k)+1) = O(3^kk^3)$.

\begin{theorem}
There is an FPT algorithm of running time 
$3^{O(k^4)}n^{O(1)}$ for \prvcew parameterized by the vertex cover number $k$.
\end{theorem}

\begin{proof}
Let $(G,w,k)$ be the input instance. 
We find a minimum vertex cover $S$ of size $k$\cite{harris2024faster}. We partition vertices of $I$ into equivalence classes. Now, we apply the reduction rule exhaustively and get an equivalent instance $(H^{\star},w,k)$. The number of unweighted edges in $H^{\star}$ is at most $O(3^k k^3)$. If $(H^{\star},w,k)$ is a yes-instance, then by Lemma~\ref{lem:precolor}, there is a weight function $w'$ such that  ${\sf color}_{w'}(v)-base_w(v)  \leq 8k^2+8k$. This implies that the number of unweighted edges assigned $1$ by $w'$ is at most $k(8k^2+8k)$. So we choose all edge subsets of unweighted edges of $H^{\star}$ of size at most $k(8k^2+8k)$ and test in linear time whether setting those edges to $1$, and the remaining to $0$ is a proper weight function. 
The number of such edge subsets of size at most $k(8k^2+8k)$ in $E(H^{\star})$ is at most $2^{O(k^4)}$. Thus, the total running time follows.   
\end{proof}

\section{Conclusion}
\label{sec:conclusion}

In this work, we studied 
Vertex-Coloring $\{0,1\}$-Edge-Weighting problem 
and its generalization, Vertex-Coloring Pre-edge-Weighting from the parameterized complexity perspective. We considered the parameters -- vertex cover number, feedback vertex set number and treewidth of the input graph. We proved that the base problem and the pre-weighted version when all the pre-weights are 1 are FPT parameterized by the vertex cover number. An interesting open problem is the parameterized complexity of pre-weighted version of the problem with parameter vertex cover number. 



\bibliography{intro,ref}

\end{document}